\documentclass[12pt]{article}
\usepackage{amsmath}
\usepackage{amsfonts}
\usepackage{amssymb}
\usepackage{amsthm}
\usepackage{geometry}
\usepackage[swedish, english]{babel}
\usepackage{pdfpages}
\usepackage{pdflscape}
\usepackage{setspace}
\usepackage[utf8]{inputenc}
\usepackage{bm, upgreek}
\usepackage{booktabs}
\usepackage{threeparttable}
\usepackage{floatrow}
\usepackage{natbib}
\usepackage{enumerate}
\usepackage{graphicx}
\theoremstyle{plain}
\newtheorem{theorem}{Theorem}
\newtheorem*{proof*}{Proof}

\newtheorem{remark}{Remark}
\newtheorem{condition}{Condition}

\newtheorem{lemma}{Lemma}

\theoremstyle{definition}
\newtheorem*{definition}{Definition of a design}

\usepackage{tikz}
\usepackage{pgfplots}

\newcommand*{\Perm}[2]{{}^{#1}\!P_{#2}}%
\newcommand{\var}{\text{Var}}
\newcommand{\mse}{\text{MSE}}

\pgfplotsset{compat=1.13}

\usetikzlibrary{decorations.pathreplacing}
\tikzset{
mybrace/.style={decorate,decoration={brace,aspect=#1}}
}

\usetikzlibrary{backgrounds}
\usetikzlibrary{shapes,decorations,arrows,calc,arrows.meta,fit,positioning}

\title{Properties of restricted randomization with implications for experimental design\footnote{We are grateful for helpful comments and suggestions by Guido Imbens and Per Johansson.}}
\author{Mattias Nordin\footnote{Department of Statistics and Uppsala Centre for Fiscal Studies, Uppsala University. Corresponding author, email: Mattias.Nordin@statistics.uu.se} and M\aa rten Schultzberg\footnote{Department of Statistics, Uppsala University}}
\begin{document}

\maketitle
\begin{abstract}
Recently, there as been an increasing interest in the use of heavily restricted randomization designs which enforces balance on observed covariates in randomized controlled trials. However, when restrictions are strict, there is a risk that the treatment effect estimator will have a very high mean squared error. In this paper, we formalize this risk and propose a novel combinatoric-based approach to describe and address this issue. First, we validate our new approach by re-proving some known properties of complete randomization and restricted randomization. Second, we propose a novel diagnostic measure for restricted designs that only use the information embedded in the combinatorics of the design. Third, we show that the variance of the mean squared error of the difference-in-means estimator in a randomized experiment is a linear function of this diagnostic measure. Finally, we identify situations in which restricted designs can lead to an increased risk of getting a high mean squared error and discuss how our diagnostic measure can be used to detect such designs. Our results have implications for any restricted randomization design and can be used to evaluate the trade-off between enforcing balance on observed covariates and avoiding too restrictive designs.
\end{abstract}
 
\noindent%
\vfill
  
\newpage
\onehalfspace
\section{Introduction}
For a long time, the gold standard of causal inference has been randomized experiments. By randomly selecting one group of individuals to be treated and one group to serve as controls, it is guaranteed that the two groups are comparable \emph{in expectation}. However, for a realized treatment assignment, there is no guarantee--and in fact, it is generally not the case--that the two groups are similar in both observed and unobserved characteristics. Therefore, it is possible that an estimated treatment effect is substantially different from the true treatment effect. 

The traditional solution to this problem has been to use stratification, or blocked randomization, where randomization is performed within strata based on observed discrete (or discretized) covariates, thereby ensuring balance on these characteristics. However, with many, possibly continuous, covariates, there has not existed a straightforward solution to the problem of imbalances between the treatment and control groups.

Recently, due to vastly increasing computing power, a number of methods have emerged to tackle this problem. For instance, \cite{Morgan2012} suggest that it is possible to perform a large number of randomizations and pick a treatment assignment with a very small imbalance between the treatment group and control group (rerandomization). Furthermore, there are several papers who have developed frameworks and algorithms for finding ``optimal'' or ``near-optimal'' designs (\citealt{Lauretto2017, Kallus2018, Krieger2019}). Usually, such methods are based on restricting the set of treatment assignments so narrowly that only assignments with minimal imbalances in observed covariates, according to some criterion, are considered in the randomization.

However, by heavily restricting the set of admissible treatment assignments one also takes a greater risk of getting a ``bad'' design. \cite{efron_forcing_1971} and \cite{Wu1981} show that different versions of complete randomization (the unrestricted design loosely defined as the design where all possible treatment assignments are equally likely to be selected) minimizes the maximum mean squared error (MSE) of the treatment effect estimator that can be reached, the so-called \emph{minimax property of complete randomization}.

Just as \cite{efron_forcing_1971} and \cite{Wu1981}, we study the risk of getting a high MSE for different designs, but we do it in a different framework. While the aforementioned papers consider situations with the worst possible data that can occur (i.e., in a game against nature), we work in a framework that is independent on the data. 

Instead, we introduce the notion of the experimenter being behing a ``veil of ignorance'' where data cannot be observed. In such a case, the experimenter can be viewed as randomly selecting a design from a set of indisinguishable designs (a notion that we formalize below). While each design has its own MSE, this information is unobserved to the experimenter behind the veil of ignorance, and we show that the largest possible MSE that can be reached is minimized under complete randomization. This result is analogous to the minimax-property from \cite{efron_forcing_1971} and \cite{Wu1981}, but is proven in a framework without any assumption of a specific data-generating process.

The advantage with our framework is that we can go further and show that not only is complete randomization minimax optimal, but that the heavier a design is restricted (in the sense that the design contains fewer assignment vectors) the higher the MSE is in the worst possible case. Furthermore, the experimenter might not only be interested in minimizing the maximum MSE, but may consider other measures of risk (see, e.g., \citealt{kapelner2019} and \citealt{kapelner_harmonizing_2020}).\footnote{\cite{kapelner2019} and \cite{kapelner_harmonizing_2020}  consider a ``tail criterion'' where they are interested in the MSE of a quantile of the worst possible mean squared errors, such as the 5\% worst cases. They argue that instead of only focusing on the absolute worst case--a case which may almost never happen--it makes sense to instead consider the bad cases that happens with at least some frequency.} In our framework, a natural measure of risk is the variance of the MSE. We show that not only does the maximum MSE increase when designs become more restricted, but so does the variance of the MSE. However, we also show that for a given restriction, even behind the veil of ignorance it is possible to identify designs with a lower variance of the MSE. Specifically, we introduce the \emph{assignment correlation}--a measure of how correlated different assignment vectors in a design are with each other--and show that the variance of the MSE increases linearly in this measure. Intuitively, a design with a high assignment correlation contains less unique information, implying an increased risk of getting a large MSE.\footnote{Consistent with our result, in a recent paper, \cite{krieger_improving_2020} find that the power of the randomzation test becomes worse when assignment vectors are highly correlated with each other.} Fortunately, the assignment correlation is straightforward to calculate as it only depends on the combinatoric relationship of the assignment vectors.

The practical value therefore lies in that the experimenter can readily observe the assignment correlation. We show that for traditional designs, such as block randomization, the assignment correlation is bounded and cannot take on extreme values. However, more modern designs---which tries to enforce balance on continuous covariates---have the ability to search through millions of assignment vectors, and may also use algorithms that cleverly find a narrow set of admissible assignment vectors, thereby heavily restricting the design. In such cases there is a real possibility of getting a large assignment correlation.

Through a simple simulation study, we show that such algorithmic designs will in most cases work well. However, we also show that there are instances when they break down and yield very high assignment correlations. This phenomenon is most likely to occur when data contain outliers which forces some units to always be treated together.

While the theoretical result is shown under the veil of ignorance, we argue that the assignment correlation is a relevant measure of risk also in situations when the experimenter observes covariates likely to affect the outcome. For instance, with rerandomization, there is not currently any agreed upon guideline for how strictly balance should be enforced (see, e.g., discussions in \citealt{Morgan2012} and \citealt{Li2018}). On the one hand, the stricter balance is enforced, the lower the expected MSE is if covariates are informative in explaining the outcome. At the same time, with very strict designs, the ``randomness'' of the small set of remaining assignment vectors may be compromized, in the sense that these assignment vectors are very similar to each other. By observing a high assignment correlation in such a case, the experimenter is made aware of this potential issue and may take appropriate action. We elaborate more on this point in the discussion at the end of the paper.

In the next section we lay out the framework for restricted randomization that we work in. In Section \ref{sec:minimax} we re-prove the minimax property of complete randomization, whereas we in Section \ref{sec:comb_u} introduce the assignment correlation and prove that the variance of the MSE is a linearly increasing function of it. After elaborating on some of the properties of the assignment correlation, in Section \ref{sec:implications} we discuss implications for some common designs and perform a simple simulation study. Finally, in Section \ref{sec:discussion}, we discuss the practical implications for experimental designs in light of the theoretical concept of ``veil of ignorance'' and we also give advice for how the assignment correlation can be used in practice.

\section{Theory of restricted randomization designs}\label{sec:theory}
We consider an experimental setting where a sample of $N$ units is observed and we want to estimate the effect of some intervention (treatment) relative to some baseline (control). The outcome of interest is denoted $Y$, with $Y_i(1)$ being the potential outcome for unit $i$ if treated and $Y_i(0)$ the potential outcome for unit $i$ if not treated. This formulation is general in the sense that we allow for heterogeneous treatment effects. The interest is in estimating the sample average treatment effect, $\tau:=1/N\sum_{i=1}^N (Y_i(1) - Y_i(0))$. 

The division of units into treatment and control is made by a $N\times 1$ \emph{assignment vector}, $\mathbf{w}\in \mathcal{W}$, containing zeros (untreated) and ones (treated). For simplicity, we will only consider assignment vectors where the sample is split into two evenly sized groups (termed a forced balanced procedure by \citealt{rosenberger_randomization_2015}), which means that the cardinality of $\mathcal{W}$ is $|\mathcal{W}|=\binom{N}{N/2}=N_A$.

\begin{definition}
	We define a \emph{design} as a set of assignment vectors from which one assignment vector is randomly chosen. A design has to satisfy the mirror-property (\citealt{Johansson2018a,kapelner_harmonizing_2020}) which says that if assignment vector $\mathbf{w}$ is included in a given design, then the assignment vector $\mathbf{1-w}$ must also be included in that design. A given design is denoted $\mathcal{W}_H^k\subseteq \mathcal{W}$ where $H$ is the cardinality of the design and $k=1,\ldots,\binom{N_A/2}{H/2}$ indexes the different designs that are possible for a given $H$.
\end{definition}

We enforce the mirror-property to guarantee that the difference-in-means estimator is unbiased.  Note that the number of designs possible for a given $H$ is $\binom{N_A/2}{H/2}$ and not $\binom{N_A}{H}$, because of the the mirror-property.

 Let $\mathcal{K}$ be the set of all possible designs, implying that the cardinality of $\mathcal{K}$ is
\begin{equation}
|\mathcal{K}|=\sum_{H_2=1}^{N_A/2} \binom{N_A/2}{H_2} = 2^{N_A/2}-1,
\end{equation}
where $H_2=H/2$. Finally, we define $\mathcal{K}_H\subseteq \mathcal{K}$ as the set containing all designs of cardinality $H$.

The difference-in-means estimator of $\tau$ for an assignment vector $\mathbf{w}_j$ is
\begin{equation}
 	\hat{\tau}_j=\frac{1}{N/2}\left(\mathbf{w}_j^T\mathbf{Y}(1) - (\mathbf{1}- \mathbf{w}_j)^T\mathbf{Y}(0)\right),
\end{equation}
where $\mathbf{Y}(0)$ is a $N\times 1$ vector of potential outcomes if not treated and $\mathbf{Y}(1)$ the corresponding potential outcomes if treated. For a given design of size $H$, $\mathcal{W}_H^k$, there are $H$ estimates that can be obtained. Because of the mirror-property, the difference-in-means estimator is unbiased:
\begin{equation}
	\frac{1}{H}\sum_{\mathbf{w}_j\in \mathcal{W}_H^k} \hat\tau_j = \tau,
\end{equation}
The mean squared error (MSE) of the difference-in-means estimator for a design $\mathcal{W}_H^k \in \mathcal{K}_H$ is
\begin{equation}
\mse(\hat\tau_{\mathcal{\{W}_H^k\}}):=\frac{1}{H}\sum_{\mathbf{w}_j\in \mathcal{W}_H^k} (\hat{\tau}_j - \tau)^2.
\end{equation}
Note that because the estimator is unbiased, the MSE and the variance are identical.

A randomized design for which $H<N_A$, is sometimes called a \textit{restricted randomization design}\footnote{Frank Yates in his discussion in \cite{yates1948} originally used the term \emph{restricted} to refer to restrictions made to latin square designs in the context of agricultural experiments. Later on, \cite{youden_randomization_1972} used the term \emph{constrained randomization} also in the context of agricultural experiments.} whereas we call the case when $H=N_A$ \emph{complete randomization}.\footnote{In contexts where designs do not have to be balanced, a forced balanced design may also be called a restricted design. The case when treatment assignment is completely random (i.e., the number of treated and control units need not be the same) is sometimes called \emph{completely randomized design} (CRD), whereas the case of a forced balanced design may be called \emph{balanced completly randomized design} (BCRD), see e.g., \cite{Wu1981,kapelner_harmonizing_2020}. For simplicity, because we exclude CRD, we label BCRD as \emph{complete randomization}.} Examples of restricted designs include block designs and rerandomization designs, both of which imply a value of $H<N_A$. The purpose of such restrictions is most often to reduce the MSE of the estimator. Restricted randomization designs generally try to achieve this by excluding assignment vectors from $\mathcal{W}$ that have large imbalances between treatment and control groups in observed covariates, where these covariates are thought to be related to the potential outcomes.

However, as we show below, a restricted design also comes with some inherent risk that the selected design imply a higher MSE than under complete randomization. To formalize this risk, we make use of the idea of being behind a ``veil of ignorance'' where we have no \emph{a priori} knowledge of which assignment vectors are more likely to produce estimates closer to $\tau$ than any other assignment vectors. Another way of framing this idea is to say that we study different restricted designs without having observed any data.

As an example, consider the case with an experiment involving $N/2$ men and $N/2$ women and we want to use a block design to ensure that an equal number of men and women are in the treatment and control groups (i.e., there are two blocks, one for each gender). The value of $H$ in this case is $\binom{N/2}{N/4}^2<N_A$. However, behind the veil of ignorance we do not observe which units are men and which are women.\footnote{While it may seem odd to say that we do not observe this, note that it would be equivalent to the case where we do observe gender, but gender does not relate to the outcome. We elaborate more on this point in Section \ref{sec:discussion}.} In such a case, a block design which balances gender could be thought of as being randomly chosen from the set of all $\binom{N_A/2}{H/2}=\binom{N_A/2}{\binom{N/2}{N/4}^2/2}$ possible block designs with two equal-sized blocks.

The worst possible case is that data is in such a way that we select the design from the set which has the highest MSE. This ``worst case'' is similar to what \cite{efron_forcing_1971} and \cite{Wu1981} consider in their derivations of the minimax property of complete randomization, with the difference being that we do not assume any specific data-generating process. In the next section, we show how the minimax property of complete randomization can be proven in our framework.

\subsection{Re-proving the minimax property of complete randomization using the combinatoric approach \label{sec:minimax}}
The main result in this paper, theorem \ref{thm:varvar_u}, is proved by studying the behavior of the difference-in-means estimator across various restricted designs. In this section, we build intuition for this way of studying restricted designs by re-proving some known properties of restricted randomization designs. Specifically, we provide an alternative proof that complete randomization satisfies the minimax-property (i.e., that it has the smallest maximum MSE) and show that the more restricted a design is, the larger the maximum MSE and the variance of the MSE are.

The key to our formulation of the restricted randomization design is to consider various subsets of the full set of designs, $\mathcal{K}$. Because the number of possible designs grows extremely fast in $N$, we will show a simple example for $N=4$ to build intuition.\footnote{The number of designs is $2^{N_A/2}-1,$, which for $N=2,4,6,8,10\ldots$ equals $1, 7, 1023, 3.4\cdot10^{10}, 8.5\cdot 10^{37}\ldots$} For $N=4$ there are $N_A = \binom{4}{2} = 6$ different assignment vectors:
\begin{equation}
	\mathbf{w}_1 =
	\left[
		\begin{array}{c}
			1 \\ 1 \\ 0 \\ 0
		\end{array}
	\right],
	\mathbf{w}_2 =
	\left[
		\begin{array}{c}
			1 \\ 0 \\ 1 \\ 0
		\end{array}
	\right],
	\mathbf{w}_3 =
	\left[
		\begin{array}{c}
			1 \\ 0 \\ 0 \\ 1
		\end{array}
	\right],
	\mathbf{w}_4 =
	\left[
		\begin{array}{c}
			0 \\ 1 \\ 1 \\ 0
		\end{array}
	\right],
	\mathbf{w}_5 =
	\left[
		\begin{array}{c}
			0 \\ 1 \\ 0 \\ 1
		\end{array}
	\right],
	\mathbf{w}_6 =
	\left[
		\begin{array}{c}
			0 \\ 0 \\ 1 \\ 1
		\end{array}
	\right].
\end{equation}
For each assignment vector, there is an associated estimate of $\tau$, $\hat\tau_1,\ldots,\hat\tau_6$. For $H=2$ and $H=4$ there are three different designs each that satisfy the mirror property, whereas for $H=6$ there is one:
\begin{align}
	\mathcal{W}_2^1 &= \left\{\mathbf{w}_1, \mathbf{w}_6\right\},\quad
	\mathcal{W}_2^2 = \left\{\mathbf{w}_2, \mathbf{w}_5\right\},\quad
	\mathcal{W}_2^3 = \left\{\mathbf{w}_3, \mathbf{w}_4\right\}, \nonumber \\
	\mathcal{W}_4^1 &= \left\{\mathbf{w}_1, \mathbf{w}_2, \mathbf{w}_5, \mathbf{w}_6\right\},\quad
	\mathcal{W}_4^2 = \left\{\mathbf{w}_1, \mathbf{w}_3, \mathbf{w}_4, \mathbf{w}_6\right\},\quad
	\mathcal{W}_4^3 = \left\{\mathbf{w}_2, \mathbf{w}_3, \mathbf{w}_4, \mathbf{w}_5\right\}, \nonumber \\
	\mathcal{W}_6^1 &= \left\{\mathbf{w}_1, \mathbf{w}_2, \mathbf{w}_3, \mathbf{w}_4, \mathbf{w}_5, \mathbf{w}_6\right\}.
\end{align}

Each design has an associated MSE of $\hat\tau$, $\mse(\hat\tau_{\mathcal{\{W}_H^k\}})$. Suppose we have the following data:
\begin{equation}
	\mathbf{Y}(0) = 
	\left[
		\begin{array}{c}
			1 \\ 2 \\ 3 \\ 4
		\end{array}
	\right], \quad
	\mathbf{Y}(1) = \mathbf{Y}(0) +
	\left[
		\begin{array}{c}
			3 \\ 4 \\ -2 \\ 3
		\end{array}
	\right] =
	\left[
		\begin{array}{c}
			4 \\ 6 \\ 1 \\ 7
		\end{array}
	\right].
\end{equation}
The associated mean squared errors are
\begin{align}
	\mse(\hat\tau_{\{\mathcal{W}_2^1\}}) &= \frac{2}{8},\quad \mse(\hat\tau_{\{\mathcal{W}_2^2\}}) = \frac{50}{8},\quad \mse(\hat\tau_{\{\mathcal{W}_2^3\}}) = \frac{8}{8}, \nonumber \\
	\mse(\hat\tau_{\{\mathcal{W}_4^1\}}) &= \frac{26}{8},\quad \mse(\hat\tau_{\{\mathcal{W}_4^2\}}) = \frac{5}{8},\quad \mse(\hat\tau_{\{\mathcal{W}_4^3\}}) = \frac{29}{8}, \nonumber \\
	\mse(\hat\tau_{\{\mathcal{W}_6^1\}}) &= \frac{20}{8}. \label{eq:example_var}
\end{align}
In this section, we consider the distribution of MSE for different sets $\mathcal{K}_H$, i.e., the distribution of MSE for all the ways $H$ assignment vectors can be drawn from the total of $N_A=\binom{N}{N/2}$ assignment vectors. Naturally, there is only one way to draw $N_A$ assignment vectors from $N_A$. That is, there is only one design containing all assignment vectors: \emph{complete randomization}. The MSE for this design is
\begin{equation}
 	\mse(\hat\tau_{\{\mathcal{W}_{N_A}^1\}})=\frac{1}{N_A} \sum_{j=1}^{N_A} (\hat\tau_{j} - \tau)^2=\sigma^2_{CR}. \label{eq:var_cr}
\end{equation}
For the example above, $\mse(\hat\tau_{\{\mathcal{W}_6^1\}}) = \sigma^2_{CR} = \frac{20}{8}$. 

For the mean squared errors in equation \eqref{eq:example_var}, it is the case that for each $\mathcal{K}_H$ (i.e., for $H=2, 4, 6$), the average MSE is equal to $\sigma^2_{CR}$. This fact is not a coincidence but something that can be generalized under the following condition: 
\begin{condition}\label{cond:cond1}
	$\tilde{\mathcal{K}}_H=\{\tilde{\mathcal{W}}_H^1,\dots,\mathcal{\tilde W}_H^m\}\subseteq \mathcal{K}_H$ is a set of designs such that
	\begin{equation}
		\sum_{k=1}^m 1\left[\mathbf{w}\in\mathcal{\tilde W}_H^k\right] = c, \quad \forall \mathbf{w}\in\mathcal{W},
	\end{equation}
	where $1[\cdot]$ is the indicator function (taking the value of one if the statement in brackets is true, and zero otherwise) and $c$ is a constant.
\end{condition}
This condition says that the number of times an assignment vector occurs over all designs in $\tilde{\mathcal{K}}_H$ is the same for all assignment vectors in $\mathcal{W}$. For the example above, $c=1,2,1$ for $\tilde{\mathcal{K}}_H=\mathcal{K}_2,\mathcal{K}_4,\mathcal{K}_6$. The motivation behind condition \ref{cond:cond1} is that we are studying the behavior of different designs with no information available about which specific assignment vectors are more likely to produce an estimate close to $\tau$. Therefore, there is no reason to prefer one specific assignment vector over any other assignment vector and so we only consider sets of designs where each assignment vector is equally likely to be selected. We can now state the following result:

\begin{theorem}\label{thm:mean_fix}
	Consider a set of designs $\tilde{\mathcal{K}}_H = \{\tilde{\mathcal{W}}_H^1,\ldots,\tilde{\mathcal{W}}_H^m\}\subseteq \mathcal{K}_H$ satisfying condition \ref{cond:cond1}. The expected MSE of the difference-in-means estimator for a randomly chosen design in $\tilde{\mathcal{K}}_H$ is equal to the MSE under complete randomization. I.e., 
	\begin{equation}
		\frac{1}{m}\sum_{k=1}^{m}\mse(\hat\tau_{\{\mathcal{\tilde{W}}_H^k\}})
		=\sigma^2_{CR}.
	\end{equation}
\end{theorem} 
\begin{proof}
	See the supporting information.
\end{proof}
The intuition behind this theorem is straightforward. By condition \ref{cond:cond1}, each assignment vector occurs an equal number of times over all the designs in $\tilde{\mathcal{K}}_H$. For any such set, the average MSE over all the designs in this set is therefore always going to be the same. And because each assignment vector occurs an equal number of times under complete randomization, the expected MSE is always equal to the MSE under complete randomization.
\begin{remark}
	$\mathcal{K}_H$ satisfies condition \ref{cond:cond1}, which means that the expected MSE for a randomly selected design out of all possible designs containing $H$ assignment vectors is equal to the MSE under complete randomization.
\end{remark}
The theorem therefore implies that by randomly selecting a design containing $H$ assignment vectors out of all possible such designs, the expected MSE of the difference-in-means estimator is the same as under complete randomization. However, the distribution of the mean squared errors will be different for different values of $H$.
\begin{remark}\label{rem:minmax}
  By theorem \ref{thm:mean_fix}, all sets of designs under the veil of ignorance (i.e., satisfying condition \ref{cond:cond1}) have the same expected MSE, which equals the MSE under complete randomization, Therefore, the maximum MSE for any such set of designs can never be smaller than the MSE for the set containing only complete randomization ($\mathcal{K}_{N_A}=\{\mathcal{W}\}$). Furthermore, all sets other than the set which contains only complete randomization contain more than one design. Therefore, if all difference-in-means estimates are distinct, the inequalities will in general be strict and $\mathcal{K}_{N_A}$ has the uniquely smallest maximum MSE.
\end{remark}
Remark \ref{rem:minmax} implies a version of the minimax property of complete randomization (\citealt{efron_forcing_1971,Wu1981}) which says that complete randomization minimizes the maximum MSE. In addition, we can go further by proving that more restrictive designs have a greater maximum MSE:
\begin{theorem}\label{thm:max_fix}
	For $\mathcal{K}_H$ and $\mathcal{K}_{H'}$ such that $H<H'$, it is the case that 
	\begin{enumerate}[i)]
		\item the maximum MSE in the set of all designs in $\mathcal{K}_H$ is greater than the maximum MSE in the set of all designs in $\mathcal{K}_{H'}$.
		\item the minimum MSE in the set of all designs in $\mathcal{K}_H$ is smaller than the minimum MSE in the set of all designs in $\mathcal{K}_{H'}$.
	\end{enumerate}
	If all difference-in-means estimates are distinct, the inequalities are strict.
\end{theorem}

\begin{proof}
	See the supporting information.
\end{proof}

Theorem \ref{thm:max_fix} says that the worst possible design, in terms of MSE, gets worse as $H$ decreases. That is, under the veil of ignorance, the more restrictive the design is, the larger the maximum MSE is.

To measure the risk associated with randomly selecting a design, we not only study the maximum possible MSE, but also the variance of the MSE of the difference-in-means estimator. We define this variance as
\begin{equation}
	\var\left(\mse(\hat\tau_{\{\mathcal{W}_H\}}):\mathcal{W}_H\in \mathcal{K}_H\right):=\frac{1}{\binom{N_A/2}{H/2}} \sum_{k=1}^{\binom{N_A/2}{H/2}} \left( \frac{1}{H} \sum_{\mathbf{w}_j \in \mathcal{W}_{H}^k} (\hat\tau_j-\tau)^2 -\sigma^2_{CR}\right)^2. \label{eq:def_var_mse}
\end{equation}

\begin{theorem}
	\label{thm:varvar_h}{}
	Under condition \ref{cond:cond1}, the variance of the MSE of the difference-in-means estimator is decreasing in $H$.
\end{theorem}
\begin{proof}
	See the supporting information.
\end{proof}

Theorem \ref{thm:varvar_h} reinforces the lesson from theorem \ref{thm:max_fix} that the risk of getting a large MSE of the difference-in-means estimator is smaller the greater $H$ is. Without any knowledge about $\mathbf{Y}(0)$ and $\mathbf{Y}(1)$ it therefore seems as if it is always better to increase $H$ so as to decrease the variance of the MSE of the difference-in-means estimator. However, we have only shown that this is true for sets which contain \emph{all} designs of size $H$ (i.e., $\mathcal{K}_H$). It might seem impossible to select a subset of 
$\mathcal{K}_H$ without any \emph{a priori} information, but as the following section will show, it is in fact possible to select subsets of $\mathcal{K}_H$ where the variance of the MSE is smaller. It is possible to do so by using the combinatorial relationship between different assignment vectors.

\subsection{Combinatorial uniqueness \label{sec:comb_u}}
There is one source of information that, to the best of our knowledge, has not been utilized in experimental design, namely the information in the combinations, what we call the \emph{pairwise uniqueness} of the assignment vectors in a set. In this section, we show that the variance of the MSE of the difference-in-means estimator depends on this uniqueness, and that the uniqueness can be used as a source of information about how ``risky'' a design is behind the veil of ignorance.

We begin with a simple example that illustrates what we mean with the concept of uniqueness. Consider the case where we have an experiment with $N=8$. If the first four units are treated and the last four are not, the corresponding assignment vector is given by $\mathbf{w} = [\begin{array}{cccccccc}1&1&1&1&0&0&0&0\end{array}]^T$. Consider a set of two assignment vectors,  
\begin{equation}
\{[\begin{array}{cccccccc}1&1&1&1&0&0&0&0\end{array}]^T, [\begin{array}{cccccccc}1&1&1&0&1&0&0&0\end{array}]^T\}
\end{equation}
In the first vector, units 1, 2, 3 and 4 are treated and in the second, units 1, 2, 3 and 5 are treated. The uniqueness of this pair of vectors is defined as the number of units that are uniquely assigned to treatment in the first but not the second vector, which in this case is one. Consider another example with the set of vectors
\begin{equation}
\{[\begin{array}{cccccccc}1&1&1&1&0&0&0&0\end{array}]^T, [\begin{array}{cccccccc}1&0&0&0&1&1&1&0\end{array}]^T\}
\end{equation}
Here, the uniqueness is three, as there are three units assigned to treatment in the first but not the second vector. The uniqueness is theoretically bounded between 0 (where an assignment vector is compared to itself) and $N/2$ (where the assignment vector is compared to its mirror).

Let $w^i_j$ be the $i$th element of assignment vector $j$. The uniqueness between two assignment vectors $j$ and $j'$ is defined as
\begin{equation}
U_{j,j'}:= \sum_{i=1}^N 1[w^i_j=1\land w^i_{j'}=0],
\end{equation}
where $1[\cdot]$ is the indicator function.

The uniqueness between two assignment vectors does not depend on data, and so it is a measure that is available even behind the veil of ignorance. Our interest is in studying whether it is possible to use this measure to reduce the risk of a large MSE of the difference-in-means estimator. 

To study this question, we continue to restrict attention to sets of designs with $H$ assignment vectors, $\tilde{\mathcal{K}}_H\subseteq \mathcal{K}_H$, which satisfies condition \ref{cond:cond1}, i.e., where all assignment vectors are equally likely to occur over all designs in the set. By theorem \eqref{thm:mean_fix}, we therefore know that the expected MSE for any sets of designs under study equals $\sigma^2_{CR}$. We now add a second condition:

\begin{condition} \label{cond:cond2}
	$\tilde{\mathcal{K}}_H=\{\tilde{\mathcal{W}}_H^1,\dots,\mathcal{\tilde W}_H^m\}\subseteq \mathcal{K}_H$ is a set of designs such that
	\begin{equation}
		\sum_{k=1}^m 1\left[\{\mathbf{w}_j, \mathbf{w}_{j'}\}\subseteq \mathcal{\tilde W}_H^k \right] = d_u, \quad \forall \mathbf{w}_j,\mathbf{w}_{j'}\in\mathcal{W}:U_{j,j'}=u,
	\end{equation}
	where $1[\cdot]$ is the indicator function and $d_u$ is a constant for a given value of $u$.
\end{condition}

This condition says that two different assignment vectors $\mathbf{w}_j$ and $\mathbf{w}_{j'}$ will occur together in the same number of designs as any other pair of assignment vectors with the same pairwise uniqueness. The intuition behind this condition is that behind the veil of ignorance, for pairs of assignment vectors with the same uniqueness there is no information available which allows us to say that one pair should occur more often than any other pair. Condition \ref{cond:cond2} therefore requires such pairs to occur the same number of times.

The insight that the pairwise uniqueness can provide additional information about the variability of the MSE leads us to the following result:
\begin{theorem}\label{thm:varvar_u}
	Consider a set of designs $\tilde{\mathcal{K}}_H = \{\tilde{\mathcal{W}}_H^1,\dots,\mathcal{\tilde W}_H^m\}\subseteq \mathcal{K}_H$ satisfying conditions \ref{cond:cond1} and \ref{cond:cond2}. The variance of the MSE of the difference-in-means estimator can be written as
	\begin{equation}
		\var\left(\mse(\hat\tau_{\{\mathcal{W}_H\}}):\mathcal{W}_H\in \mathcal{\tilde K}_H\right)=\frac{4}{N^2}\psi \left( 2\frac{N_A-H}{HN_A} + 
    	\frac{H-2}{H}\phi(\mathcal{\tilde{ K}}_H)-\frac{N_A-2}{N_A}\phi(\mathcal{ K}) \right), \label{eq:varvar_in_sample}
	\end{equation}
	where $\psi=\psi(\mathbf{Y}(0),\mathbf{Y}(1))$ is a function of the data, $\phi(\tilde{\mathcal{K}}_H)$ is the expected value of $\left(4/N\right)^2(u-N/4)^2$ over all designs in $\tilde{\mathcal{K}}_H$ and $\phi(\mathcal{K})$ is the expected value of $\left(4/N\right)^2(u-N/4)^2$ over all $N_A(N_A-2)$ possible pairwise combinations of two distinct assignment vectors (excluding mirrors).
\end{theorem}
\begin{proof}
	See the supporting information.
\end{proof}
The parameter $\phi(\tilde{\mathcal{K}}_H)$ is written as
\begin{equation}
	\phi(\tilde{\mathcal{K}}_H):=\left(\frac{4}{N}\right)^2\sum_{u=1}^{N/2-1}\nu_u(\mathcal{\tilde K}_H)(u-N/4)^2, \label{eq:def_phi}
\end{equation}
where $\nu_u(\mathcal{\tilde K}_H)$ is the proportion of pairwise uniqueness values (excluding mirrors) in the designs in $\tilde{\mathcal{K}}_H$ with uniqueness $U=u$ and $\left(4/N\right)^2$ is a normalizing constant. For much of the discussion below, it will be useful to discuss the corresponding parameter for a given design, $\mathcal{W}_H$, defined as
\begin{equation}
	\varphi(\mathcal{W}_H) := \left(\frac{4}{N}\right)^2 \sum_{u=1}^{N/2-1} \mu_u(\mathcal{W}_H) (u-N/4)^2,
\end{equation}
where $\mu_u(\mathcal{W}_H)$ is the proportion of pairwise uniqueness values (excluding mirrors) in design $\mathcal{W}_H$. It is straightforward to show that
\begin{equation}
	\phi(\tilde{\mathcal{K}}_H) =\frac{1}{m} \sum_{k=1}^m \varphi(\mathcal{\tilde W}_H^k).
\end{equation}
That is, $\phi(\tilde{\mathcal{K}}_H)$ is the average value of $\varphi(\mathcal{W}_H)$ over the designs in $\tilde{\mathcal{K}}_H$.

The parameter $\varphi$ is a key parameter bounded between zero (if $U=N/4$ for all pairs) and one.\footnote{Technically, the upper bound is $\left(\frac{4}{N}\right)^2(N/4-1)^2=1-\frac{8}{N}+\frac{16}{N^2}<1$, which happens if $U$ equals 1 or $N/2-1$ for all pairs.} As $\varphi$ increases, the less unique information is available in each assignment vector in the design. At first glance, one might have expected that information should be increasing in uniqueness. However, because mirrors are always included, this is not the case. For two assignment vectors $\mathbf{w}_j,\mathbf{w}_{j'}$ with $U_{j,j'}=u$, the mirrors are also included with associated values of $U_{N_A+1-j,N_A+1-j'}=u$ and $U_{j,N_A+1-j'}=U_{N_A+1-j,j'}=N/2-u$. The pairwise uniquenesses for a design will therefore always be symmetrically distributed around $N/4$.

A large value of $\varphi$ implies that the assignment vectors in the design are highly correlated with each other, which also means that the associated treatment effect estimates will correlate with each other. On the other hand, with a small value of $\varphi$, the assignment vectors are less correlated, which also means that the estimates of $\tau$ have less correlation. We refer to $\varphi$ as the \emph{assignment correlation} for a design and $\phi$ as the average assignment correlation for a set of designs.

The amount of information available is maxmimized when the assignment correlation is zero (i.e., when $u=N/4$). The key insight from theorem \ref{thm:varvar_u} is that the variance of the MSE of the difference-in-means estimator is linearly increasing in the average assignment correlation, i.e., the average value of $(u-N/4)^2$ over all pairs in all designs in $\tilde{\mathcal{K}}_H$. Theorem \ref{thm:mean_fix} says that the expected MSE from a randomly selected design satisfying condition \ref{cond:cond1} is equal to $\sigma^2_{CR}$. But for the given design actually selected, the MSE is in general something different. With a large average assignment correlation, the risk that the MSE will be much larger than $\sigma^2_{CR}$ is larger than with a small average assignment correlation.

Theorem \ref{thm:varvar_u} also shows that the data, $(\mathbf{Y}(0),\mathbf{Y}(1))$, only enters multiplicatively in one place through $\psi$,\footnote{The explicit form for $\psi$ is given in the supporting information in the proof of theorem \ref{thm:varvar_u}.} which means that the relative variance of the MSE for two different sets of designs is independent of the data. Instead, what determines the relative variances are the values of $H$ and $\phi(\tilde{\mathcal{K}}_H)$ for each set of designs.

It is worth noting that for a fixed $H$, as $N$ increases, the variance of the MSE goes to zero at the rate $N^2$, which means that for large sample sizes, the variability of the MSE should be small. However, as we show in Section \ref{sec:implications}, for smaller experiments such as $N=50$, the variance of the MSE can be substantial. If $H=N_A$ (complete randomization), the variance of the MSE is zero, because $\phi(\mathcal{K})=\phi(\mathcal{K}_{N_A})$.

\subsection{Properties of the assignment correlation}

To build intuition behind how the assignment correlation affects the variance of the MSE of the difference-in-means estimator, Figure \ref{fig:hypo_tau} illustrates treatment effect estimates for two sets of designs, with each set containing two designs. Each design has $H=8$, which means that each design include eight estimates of $\tau$ with four estimates on each side of $\tau$ because mirrors are included. The estimates have been constructed such that the average MSE is the same in the two sets of designs.

In the first set of designs, the average assignment correlation is small which means that a given treatment assignment has about the same pairwise uniqueness with any other treatment assignment of, approximately, $U=N/4$. Two randomly selected assignment vectors from a design with a small assignment correlation should have treatment effect estimates that are close to uncorrelated with each other and the treatment effect estimates for such designs should therefore be close to randomly distributed around $\tau$. With a large assignment correlation on the other hand, a given treatment assignment has either a small or large pairwise uniqueness with the other assignments. The estimates with a small pairwise uniqueness therefore ``bunch together'', whereas those with large pairwise uniqueness tend to be located on opposite sides of $\tau$. The resulting MSE is therefore more variable for the set of designs with comparatively larger assignment correlations.

\begin{figure}
	\begin{tikzpicture}
		\newcommand\ra{5.5}
		\newcommand\rb{4}
		\newcommand\rc{1.5}
		\newcommand\rd{0}
		\node[text width = 3.5cm] at (1.5,\ra + 1) {Small average assignment correlation};	
		\node[text width = 3.5cm] at (1.5,\rc + 1) {Large average assignment correlation};
		\newcommand\taw{4}
		\newcommand\tbw{3.5}
		\newcommand\tcw{2}
		\newcommand\tdw{9}
		\newcommand\tav{1.8}
		\newcommand\tbv{5.4}
		\newcommand\tcv{10}
		\newcommand\tdv{8}
		\newcommand\tay{5.5}
		\newcommand\tby{6}
		\newcommand\tcy{4.8}
		\newcommand\tdy{6.5}
		\newcommand\tax{2}
		\newcommand\tbx{3}
		\newcommand\tcx{3.7}
		\newcommand\tdx{1.5575741}
		\newcommand\tamw{14 - \taw}
		\newcommand\tbmw{14 - \tbw}
		\newcommand\tcmw{14 - \tcw}
		\newcommand\tdmw{14 - \tdw}
		\newcommand\tamv{14 - \tav}
		\newcommand\tbmv{14 - \tbv}
		\newcommand\tcmv{14 - \tcv}
		\newcommand\tdmv{14 - \tdv}
		\newcommand\tamy{14 - \tay}
		\newcommand\tbmy{14 - \tby}
		\newcommand\tcmy{14 - \tcy}
		\newcommand\tdmy{14 - \tdy}
		\newcommand\tamx{14 - \tax}
		\newcommand\tbmx{14 - \tbx}
		\newcommand\tcmx{14 - \tcx}
		\newcommand\tdmx{14 - \tdx}
		\draw[thick,->] (0,\ra) -- (14,\ra);
		\draw[thick] (7,\ra+0.2) -- (7,\ra-0.2) node[below] {$\tau$};
		\draw[thick, color=black] (\taw,\ra+0.2) -- (\taw,\ra-0.2) node[below] {$\hat\tau_1$};
		\draw[thick, color=black] (\tamw,\ra+0.2) -- (\tamw,\ra-0.2) node[below] {$\hat\tau_8$};
		\draw[thick, color=black] (\tbw,\ra+0.2) -- (\tbw,\ra-0.2) node[below] {$\hat\tau_2$};
		\draw[thick, color=black] (\tbmw,\ra+0.2) -- (\tbmw,\ra-0.2) node[below] {$\hat\tau_7$};
		\draw[thick, color=black] (\tcw,\ra+0.2) -- (\tcw,\ra-0.2) node[below] {$\hat\tau_3$};
		\draw[thick, color=black] (\tcmw,\ra+0.2) -- (\tcmw,\ra-0.2) node[below] {$\hat\tau_6$};
		\draw[thick, color=black] (\tdw,\ra+0.2) -- (\tdw,\ra-0.2) node[below] {$\hat\tau_4$};
		\draw[thick, color=black] (\tdmw,\ra+0.2) -- (\tdmw,\ra-0.2) node[below] {$\hat\tau_5$};	
		\draw[thick,->] (0,\rb) -- (14,\rb);
		\draw[thick] (7,\rb+0.2) -- (7,\rb-0.2) node[below] {$\tau$};
		\draw[thick, color=black] (\tav,\rb+0.2) -- (\tav,\rb-0.2) node[below] {$\hat\tau_1$};
		\draw[thick, color=black] (\tamv,\rb+0.2) -- (\tamv,\rb-0.2) node[below] {$\hat\tau_8$};
		\draw[thick, color=black] (\tbv,\rb+0.2) -- (\tbv,\rb-0.2) node[below] {$\hat\tau_2$};
		\draw[thick, color=black] (\tbmv,\rb+0.2) -- (\tbmv,\rb-0.2) node[below] {$\hat\tau_7$};
		\draw[thick, color=black] (\tcv,\rb+0.2) -- (\tcv,\rb-0.2) node[below] {$\hat\tau_3$};
		\draw[thick, color=black] (\tcmv,\rb+0.2) -- (\tcmv,\rb-0.2) node[below] {$\hat\tau_6$};
		\draw[thick, color=black] (\tdv,\rb+0.2) -- (\tdv,\rb-0.2) node[below] {$\hat\tau_4$};
		\draw[thick, color=black] (\tdmv,\rb+0.2) -- (\tdmv,\rb-0.2) node[below] {$\hat\tau_5$};
		\draw[thick,->] (0,\rc) -- (14,\rc);
		\draw[thick] (7,\rc+0.2) -- (7,\rc-0.2) node[below] {$\tau$};
		\draw[thick, color=black] (\tay,\rc+0.2) -- (\tay,\rc-0.2) node[below] {$\hat\tau_1$};
		\draw[thick, color=black] (\tamy,\rc+0.2) -- (\tamy,\rc-0.2) node[below] {$\hat\tau_8$};
		\draw[thick, color=black] (\tby,\rc+0.2) -- (\tby,\rc-0.2) node[below] {$\hat\tau_2$};
		\draw[thick, color=black] (\tbmy,\rc+0.2) -- (\tbmy,\rc-0.2) node[below] {$\hat\tau_7$};
		\draw[thick, color=black] (\tcy,\rc+0.2) -- (\tcy,\rc-0.2) node[below] {$\hat\tau_3$};
		\draw[thick, color=black] (\tcmy,\rc+0.2) -- (\tcmy,\rc-0.2) node[below] {$\hat\tau_6$};
		\draw[thick, color=black] (\tdy,\rc+0.2) -- (\tdy,\rc-0.2) node[below] {$\hat\tau_4$};
		\draw[thick, color=black] (\tdmy,\rc+0.2) -- (\tdmy,\rc-0.2) node[below] {$\hat\tau_5$};
		\draw[thick,->] (0,\rd) -- (14,\rd);
		\draw[thick] (7,\rd+0.2) -- (7,\rd-0.2) node[below] {$\tau$};
		\draw[thick, color=black] (\tax,\rd+0.2) -- (\tax,\rd-0.2) node[below] {$\hat\tau_1$};
		\draw[thick, color=black] (\tamx,\rd+0.2) -- (\tamx,\rd-0.2) node[below] {$\hat\tau_8$};
		\draw[thick, color=black] (\tbx,\rd+0.2) -- (\tbx,\rd-0.2) node[below] {$\hat\tau_2$};
		\draw[thick, color=black] (\tbmx,\rd+0.2) -- (\tbmx,\rd-0.2) node[below] {$\hat\tau_7$};
		\draw[thick, color=black] (\tcx,\rd+0.2) -- (\tcx,\rd-0.2) node[below] {$\hat\tau_3$};
		\draw[thick, color=black] (\tcmx,\rd+0.2) -- (\tcmx,\rd-0.2) node[below] {$\hat\tau_6$};
		\draw[thick, color=black] (\tdx,\rd+0.2) -- (\tdx,\rd-0.2) node[below] {$\hat\tau_4$};
		\draw[thick, color=black] (\tdmx,\rd+0.2) -- (\tdmx,\rd-0.2) node[below] {$\hat\tau_5$};
	\end{tikzpicture}
	\caption{Hypothetical estimates for different designs \label{fig:hypo_tau}}
	\floatfoot{\footnotesize Note: The figure shows hypothetical treatment effect estimates for four different designs with $H=8$ for each design. Because mirrors are included, the estimates for each design are symmetrically distributed with average value equal to $\tau$. The average variance of the first two designs (which have a small hypothetical value of $\phi$) is identical to the average variance of the last two designs (with a large hypothetical value of $\phi$).}
\end{figure}

It is important to note that for a given design, it is straightforward to calculate $\varphi$. To do so, it is sufficient to go through all the pairs of assignment vectors from the first half of the lexicographic ordering and calculate the pairwise uniqueness. In general, the computational complexity of calculating $\varphi$ is $O(H^2)$, which means that for small $H$ it is a trivial calculation, but for large $H$, it might not be computationally feasible to calculate. In such cases, $\varphi$ can instead be Monte Carlo approximated.

By observing the assignment correlation for a design, the experimenter has an easily accessible diagnostic tool which is informative of the risk of getting a design with a large MSE. Behind the veil of ignorance, we can view $\mathcal{W}_H$ as being randomly sampled from the set of all designs with the given values of $\varphi$ and $H$. We denote that set $\mathcal{K}_{H,\varphi}$. That is, it is the case that
\begin{equation}
	\phi(\mathcal{K}_{H,\varphi}) = \varphi(\mathcal{W}_H), \quad \forall \mathcal{W}_H \in \mathcal{K}_{H,\varphi}.
\end{equation}
Behind the veil of ignorance, the experimenter can narrow the set of designs from which the design was sampled from $\mathcal{K}_H$ to $\mathcal{K}_{H,\varphi}$, and thereby make a better assessment of the risk associated with the design.

This risk is generally increasing when using a design which heavily restricts the admissible assignment vectors based on some observed covariates (something we discuss further in Section \ref{sec:implications}). By calculating the assignment correlation, the experimenter can directly observe whether this is the case. If the correlation is very large, one can choose to relax the strict criterion to allow for more diversity in the assignment vectors. On the other hand, if the correlation is not large, one can proceed with a very strict design without risking much in terms of increased variation in the MSE.

It is therefore useful to be able to know what values of the assignment correlation that can be expected. For two randomly chosen assignment vectors (excluding mirrors) it is straightforward to see that the probability mass function for $u$ is
\begin{equation}
	f(u)= \frac{\frac{N_A}{2}\binom{N/2}{u}\binom{N/2}{N/2-u}}{N_A(N_A-2)/2} = \frac{\binom{N/2}{u}\binom{N/2}{N/2-u}}{N_A-2}, \label{eq:pmf_phi}
\end{equation}
where the numerator is the number of ways to select two distinct assignment vectors with a given $u=1,\ldots,N/2-1$ and the denominator is the total number of ways of selecting two assignment vectors that are not mirrors. The expected value of $\varphi$ for any set $\mathcal{K}_H$ is the same regardless of $H$. That is, it is the case that $\phi(\mathcal{K}_H)=\phi(\mathcal{K})$ because all pairs of assignment vectors with pairwise uniqueness $u$ occurs with the same relative frequency (the parameter $\nu_u(\mathcal{K}_H)$) for any $\mathcal{K}_H$. Using equation \eqref{eq:pmf_phi}, it can be shown that
\begin{equation}
	\phi(\mathcal{K}) = \sum_{u=1}^{N/2-1} \frac{\binom{N/2}{u}\binom{N/2}{N/2-u}}{N_A-2}\left(\frac{4}{N}\right)^2(u-N/4)^2 \approx \frac{1}{N-1}, \text{ for } \mathcal{W}_H \in \mathcal{K}_H,\label{eq:exp_phi}
\end{equation}
which follows from the fact that equation \eqref{eq:pmf_phi} is (approximately) a hypergeometric probability mass function.\footnote{We have $\left(\frac{4}{N}\right)^2(u-N/4)^2 = \frac{16}{N^2}u^2 - \frac{8u}{N} + 1$. Furthermore, standard results for the hypergeometric distribution imply $E(u)=N/4$ and $E(u^2) = \frac{N^2}{16} + \frac{N^2}{16(N-1)}$. We get $\phi(\mathcal{K})\approx 1+\frac{1}{N-1}-2+1=\frac{1}{N-1}$. This result is only an approximation because we go from $u=1$ to $u=N/2-1$ instead of from $u=0$ to $u=N/2$, and we divide by $N_A-2$ instead of $N_A$. However, $\phi(\mathcal{K})-\frac{1}{N-1}$ is small even for relatively small $N$ and as $N\rightarrow \infty$, $\phi(\mathcal{K})-\frac{1}{N-1}=0$.}

While the assignment correlation is bounded between zero and one, for a given $H$, it may not be combinatorially possible to reach these boundary values. The exact distribution of the assignment correlation is complicated because of the dependence of the $\binom{H}{2}$ pairs of assignment vectors. 

\section{Implications for the design of experiments \label{sec:implications}}
In this section we use the result in the previous section to discuss the small-sample properties of the MSE. The key to our discussion is the assignment correlation. Different designs have implications for what values the assignment correlation can take and it is therefore helpful to go through some common designs.

\subsection{Block randomization \label{seq:block}} 
The perhaps most common design aside from complete randomization is block randomization. In cases with only one block covariate with two equal-sized groups (such as gender), it can be shown that the assignment correlation for such a design is
\begin{equation}
	\sum_{u=1}^{N/2-1}\frac{\sum_{i=0}^u \binom{N/4}{N/4-i}^2\binom{N/4}{N/4-(u-i)}^2}{\binom{N/2}{N/4}^2-2}\left(\frac{4}{N}\right)^2(u-N/4)^2.
\end{equation}
The assignment correlation converges to $\frac{1}{N-2}$ as $N\rightarrow \infty$ and can never take extreme values. A block design with two blocks is therefore a ``safe'' design in the sense that there are clear limitations to how large the variance of the MSE can be.
 
This result is for a block design with the fewest possible number of blocks: two. On the other end of the spectrum, there can be a block design with $N/2$ blocks with two units in each block. In that case the assignment correlation for the set of all such designs equals
\begin{equation}
	\sum_{u=1}^{N/2-1}\frac{\binom{N/2}{u}}{2^{N/2}-2}\left(\frac{4}{N}\right)^2(u-N/4)^2.
\end{equation}
In this case, the assignment correlation converges to $\frac{1}{N-N/2}=\frac{2}{N}$ as $N \rightarrow \infty$. More generally, for any block design with $N/b$ blocks with exactly $b$ units in each block, the assignment correlation converges to $\frac{1}{N-N/b}=$ as $N \rightarrow \infty$.

Overall, any given block design implies a specific assignment correlation and, different from the strategies discussed below, cannot take on extreme values.

\subsection{Rerandomization}
\cite{Morgan2012} formalize the idea of a \emph{rerandomization design} as a restricted design where only assignment vectors which satsifies a balance criterion based on some observed covariates, such as the Mahalanobis distance, are part of the design. Let $M(\mathbf{X},{\mathbf{w}})$ be the Mahalanobis distance of the mean differences of covariates $\mathbf{X}$ of treated and controls given assignment vector $\mathbf{w}$. Define a threshold value $a$ where $\Pr(M(\mathbf{X},{\mathbf{w}}) \leq a)=p_A$. The set of all admissible assignment vectors in a rerandomization design is
\begin{equation}
  \mathcal{W}_H^{RR} = \{\mathbf{w} \in \mathcal{W}: M(\mathbf{X},{\mathbf{w}}) \leq a \},
\end{equation}
where $H=N_A \cdot p_A$.

As pointed out in \cite{Schultzberg2019c}, block designs can be seen as a special case of rerandomization where the observed covariates on which to rerandomize are categorical. However, with rerandomization on continuous covariates, we can no longer be certain that the assignment correlation will stay reasonably close to $\frac{1}{N-1}$, as was the case for block randomization. Instead, the assignment correlation will depend on the distribution of the covariates.

If covariates are informative in explaining the outcome, a small $p_A$ is desirable. At the same time, when choosing $p_A$, \cite{Morgan2012} argue that care should be taken such that $p_A$ is large enough for randomization inference to be possible, whereas \cite{Li2018} warns against setting $p_A$ too small and suggest that $p_A=0.001$ is reasonable. They conclude that ``How to choose $p_A$ is an open problem'' (p. 9162).

In the framework introduced in this paper, with the help of the assignment correlation, it is possible to guide the experimenter in the choice of $p_A$. Essentially, the experimenter can keep rerandomizing, gradually lowering $p_A$, until the assignment correlation starts to increase. We elaborate on this point further in the concluding discussion.

\subsection{Algorithmic designs}
Rerandomization provides a straightforward and intuitive way of finding balanced assignment vectors for the experiment. However, the number of expected rerandomizations to ensure balance in the covariates increases exponentially with the number of covariates. That means that even with only a handful of covariates, it is computationally infeasible to rerandomize if a strict balance criterion is chosen.

In response to this issue, a number of different algorithms have been suggested which more efficiently find balanced assignment vectors, making it possible to balance a larger number of covariates (see, for instance, \citealt{Lauretto2017}, \citealt{Kallus2018} and \citealt{Krieger2019}). At the same time, for certain type of data, especially containing outliers, there may only be a few ways that treatment can be assigned which enforces balance. In that case, the assignment vectors from using algorithms may be very similar to each other, compromising the randomness of the design. With the assignment correlation, this is something that is directly observable to the experimenter. We illustrate this point below using the pair-switching algorithm of \cite{Krieger2019}.

\subsection{Simulation of assignment correlations}
For block randomization with equal-sized blocks, each design imply a specific assignment correlation that (asymptotically) is bounded between $\frac{1}{N-2}$ and $\frac{1}{N-N/2}$ (see Section \ref{seq:block}). For rerandomization and the algorithmic designs, there is no way to analytically derive the values of the assignment correlation; rather, they depend on the data. Instead, we perform a simple simulation study to analyze what values the assignment correlation can take, where we study different rerandomization strategies, as well as the pair-switching algorithm of \cite{Krieger2019}.

We let the sample size be $N=50$ and use five standard normal covariates to balance on. We consider rerandomization with $p_A=0.1, 0.01, 0.001$. For these designs, $H=p_A \cdot N_A$ which makes it computationally infeasible to exactly calculate the assignment correlation. Instead, we Monte Carlo approximate the assignment correlation by randomly sampling 5,000 assignment vectors from the first half of the lexicographic ordering (equivalent to 10,000 assignment vectors when mirrors are added). We also calculate the assignment correlation for 5,000 randomly selected assignment vectors from the first half of the lexicographic ordering (corresponding to $p_A=1$, i.e., complete randomization). For that case, equation \eqref{eq:exp_phi} tells us the analytical solution and we can therefore see the sampling error due to the Monte Carlo approximation. Finally, for the pair-switching algorithm, we randomly sample 5,000 starting vectors on which the algorithm is applied which give 5,000 new assignment vectors. For each strategy, no duplicated assignment vector is allowed: if a duplicate is found, it is discarded and a new vector is drawn.

\begin{figure}[p!]
     \includegraphics[width=\linewidth]{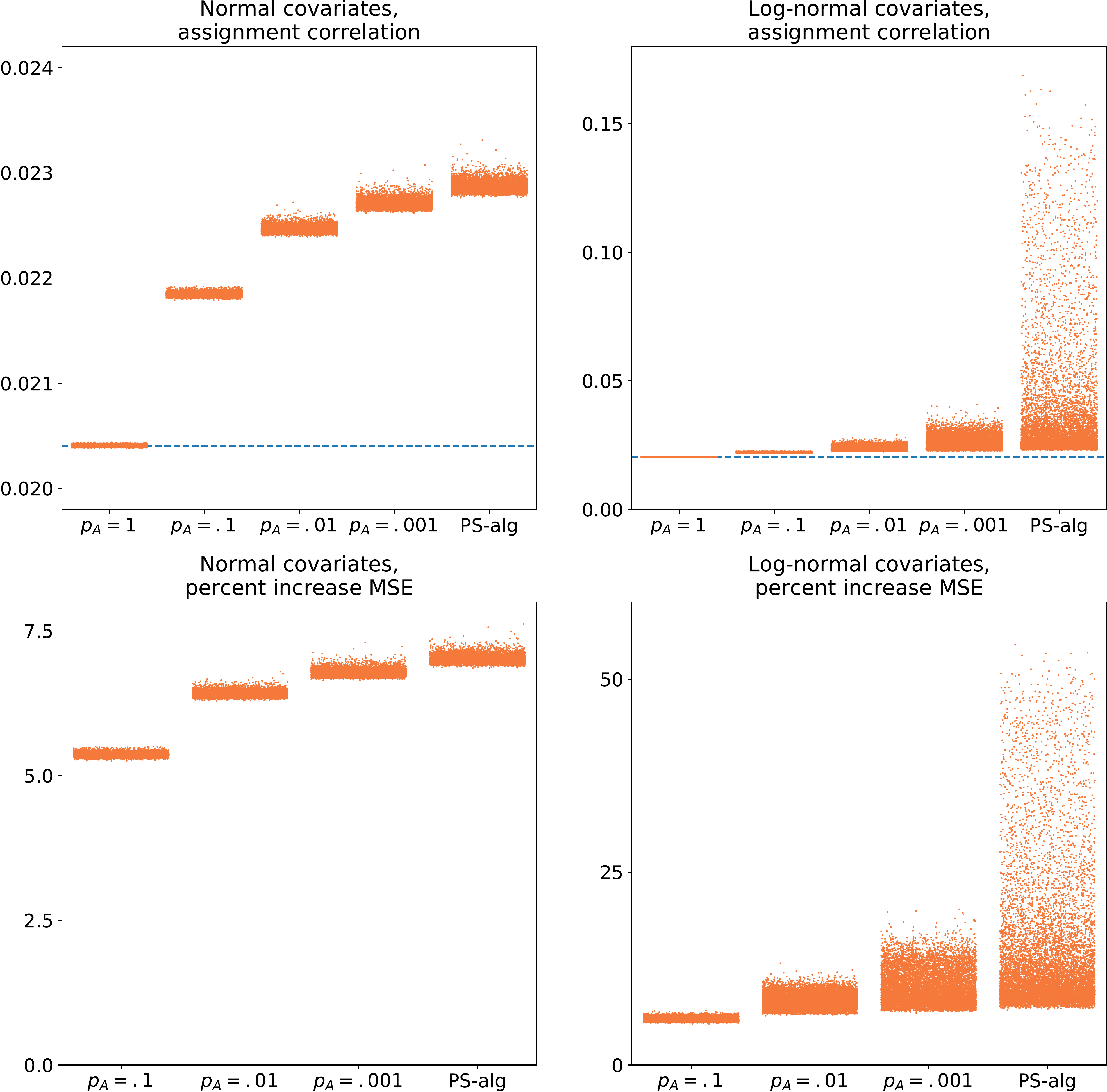}
     \caption{Distribution of assignment correlation for different designs \label{fig:dist_as_cor}}
     \floatfoot{Note: The top panel shows the distribution of the assignment correlation for different designs on normal (left graph) and lognormal (right graph) covariates with the dashed horizontal line indicating the assignment correlation for complete randomization. The bottom panel shows the corresponding MSE increase from a standard deviation increase relative to its expectation (see equation \ref{eq:rel_var_increase}). $N=50$ and for each sample and the assignment correlations are approximated with 10,000 assignment vectors (including mirrors). 10,000 random samples are drawn. Note that the scale on the $y$-axis is different for each graph.}
\end{figure}

The top left graph of Figure \ref{fig:dist_as_cor} plots the estimated assignment correlations for 10,000 different samples with Table \ref{tab:dist_as_cor} providing some key statistics from the simulation. The dashed horizontal line in the figure indicates the value of the assignment correlation under complete randomization according to equation \eqref{eq:exp_phi}. Beginning from the left, we see that for complete randomization ($p_A=1$), the Monte Carlo approximated values are very close to the true value, suggesting that 5,000 assignment vectors is sufficient to get a good estimate of the assignment correlation.

Turning to the rerandomization designs, we see that for these designs, the assignment correlation becomes larger the more restrictive the design is, but even for a fairly unrestrictive rerandomization design ($p_A=0.1$), there is a clear difference to complete randomization. The spread of the estimated assignment correlations is larger compared to complete randomization because they depend on the data: A design with $p_A=0.1$ implies a different assignment correlation for each sample, which is different for complete randomization when the assignment correlation is fixed. Finally, the pair-switching algorithm gives the largest assignment correlation. This result is expected as it is the most restrictive design with the smallest associated Mahalanobis distances.\footnote{The maximum Mahalanobis distance for each design is, on average over the 10,000 replications, 21.0 ($p_A=1$), 1.70 ($p_A=0.1$), 0.59 ($p_A=0.01$), 0.23 ($p_A=0.001$) and 0.11 (pair-switching algorithm).}

The result shows that more restrictive designs also have higher assignment correlations, but from the top left graph alone it is not possible to quantify how important this difference is. To do so, we can use equation \eqref{eq:varvar_in_sample} from theorem \ref{thm:varvar_u} which give the explicit formula for the variance of the MSE of the difference-in-means estimator. Using this information, we can quantify how variable the MSE is depending on the observed assignment correlation. If the variance of the MSE is very small compared to the MSE itself, it suggests that there is little need to worry about the small-sample behavior. If, on the other hand, there is a large degree of variability, then it becomes important to pay attention to the assignment correlation in order to guide the design choice. For instance, in the case of rerandomization, \cite{Li2018} derive the asymptotic distribution of the difference-in-means estimator which can be used to estimate the MSE. However, if the MSE varies substantially depending on the design, such an asymptotic estimator may be inappropriate as large-sample approximations are unlikely to hold.

In all the designs we study here, $H$ is a huge number. Even the restrictive pair-switching algorithm is in most cases likely to contain millions of assignment vectors when $N=50$. For large $H$, equation \eqref{eq:varvar_in_sample} simplifies to
  \begin{equation}
    \var\left(\mse(\hat\tau_{\{\mathcal{W}_H\}}):\mathcal{W}_H\in \mathcal{\tilde{K}}_H\right) = \frac{4}{N^2}\psi\left(\left( \phi(\tilde{\mathcal{K}}_H) - \frac{1}{N-1}  \right) \right),
  \end{equation}
where $\phi(\mathcal{K})$ has been replaced with $\frac{1}{N-1}$. To quantify the variability of the MSE, we will have to assume some value of $\psi$ and the expected MSE, which by theorem \ref{thm:mean_fix} equals $\sigma^2_{CR}$. If $Y(0)\sim N(0,1)$ and treatment effects are homogeneous, then the expected values for these parameters are $\psi=8$ and $\sigma^2_{CR}=\frac{4}{N}$. With this information, we can use equation \eqref{eq:varvar_in_sample} to get the variance of the MSE for different values of the assignment correlation for the set of designs $\tilde{\mathcal{K}}_{H,\varphi}$, i.e., $\phi(\tilde{\mathcal{K}}_{H,\varphi})$. Because it is relevant to compare the variance of the MSE to the expected MSE, our measure of variability is
\begin{equation}
  100 \cdot \frac{\sqrt{\var\left(\mse(\hat\tau_{\{\mathcal{W}_H\}}):\mathcal{W}_H\in \mathcal{\tilde{K}}_{H,\varphi}\right)}}{\sigma^2_{CR}} = 100 \cdot \sqrt{2\left( \phi(\tilde{\mathcal{K}}_H) - \frac{1}{N-1}  \right)}. \label{eq:rel_var_increase}
\end{equation}

That is, this measure tells us how many percent the MSE increases from a standard deviation increase relative to its expectation. The bottom left graph of Figure \ref{fig:dist_as_cor} takes the results for the assignment correlation from the top left graph as input in equation \eqref{eq:rel_var_increase} and plots the result. No result is shown for $p_A=1$ because we know that the variance of the MSE is zero for complete randomization.

We see that the average assignment correlation for $p_A=0.1$ implies that a standard deviation increase in the MSE relative to its expectation would increase the MSE with around 5.4 percent (see Table \ref{tab:dist_as_cor}). Due to the central limit theorem and the fact that the covariates are normal, we should expect the distribution of MSE to be approximately normal, which implies that 95 percent of designs with the given assignment correlation have a MSE $\pm11$ percent of expected MSE. For more restrictive designs, the variability of the MSE is higher, but because a standard deviation increase relative to its expectation increases with the square root of the assignment correlation, the additional increase in variation is modest: for the pair-switching design, the average assignment correlation imply that a standard deviation increase in the MSE relative to its expectation would increase the MSE with around 7 percent.

Overall, we see that for normal covariates, all designs imply a reasonably large variability of the MSE (for this relatively small sample size), but the variance increase from using more restrictive designs is quite small. Hence there seems to be little risk associated with very restrictive designs.

However, in this simulation study, the covariates were well-behaved whereas real data may contain covariates which are skewed and contain outliers. We therefore perform the same simulation study, but where covariates instead follow a log-normal distribution with the distribution of the assignment correlation being shown in the upper right graph of Figure \ref{fig:dist_as_cor} (with key statistics in Table \ref{tab:dist_as_cor}).

\begin{table}[p!]
\begin{threeparttable}
    \singlespace
    \begin{tabular*}{\textwidth}{@{\hskip\tabcolsep\extracolsep\fill}l*{1}{cccccccc}}\toprule
 & & & & \multicolumn{4}{c}{Quantiles}\\
\cmidrule(lr){5-8}& Mean & Min & Max & 0.5 & 0.75 & 0.975 & 0.999\\
\midrule\addlinespace\multicolumn{8}{l}{\textit{Five standard normal covariates}}\\
\addlinespace$p_A=1$ & 0.0204 & 0.0204 & 0.0204 & 0.0204 & 0.0204 & 0.0204 & 0.0204 \\
& & & & & & & \\ 
$p_A=0.1$ & 0.0218 & 0.0218 & 0.0219 & 0.0218 & 0.0219 & 0.0219 & 0.0219 \\
 & (5.4) & (5.3) & (5.5) & (5.4) & (5.4) & (5.4) & (5.5) \\
\addlinespace
$p_A=0.01$ & 0.0225 & 0.0224 & 0.0227 & 0.0225 & 0.0225 & 0.0225 & 0.0226 \\
 & (6.4) & (6.3) & (6.8) & (6.4) & (6.4) & (6.5) & (6.6) \\
\addlinespace
$p_A=0.001$ & 0.0227 & 0.0226 & 0.0231 & 0.0227 & 0.0227 & 0.0228 & 0.0229 \\
 & (6.8) & (6.6) & (7.3) & (6.8) & (6.8) & (6.9) & (7.1) \\
\addlinespace
PS-alg & 0.0229 & 0.0228 & 0.0233 & 0.0229 & 0.0229 & 0.0229 & 0.0231 \\
 & (7.0) & (6.9) & (7.6) & (7.0) & (7.0) & (7.1) & (7.3) \\
\addlinespace
\addlinespace\multicolumn{8}{l}{\textit{Five log-normal covariates}}\\
\addlinespace$p_A=1$ & 0.0204 & 0.0204 & 0.0204 & 0.0204 & 0.0204 & 0.0204 & 0.0204 \\
& & & & & & & \\ 
$p_A=0.1$ & 0.0222 & 0.0218 & 0.0229 & 0.0222 & 0.0223 & 0.0224 & 0.0227 \\
 & (6.0) & (5.3) & (7.1) & (6.0) & (6.1) & (6.4) & (6.8) \\
\addlinespace
$p_A=0.01$ & 0.0239 & 0.0224 & 0.0291 & 0.0238 & 0.0245 & 0.0254 & 0.0272 \\
 & (8.3) & (6.4) & (13.2) & (8.2) & (9.0) & (10.0) & (11.6) \\
\addlinespace
$p_A=0.001$ & 0.0258 & 0.0227 & 0.0408 & 0.0250 & 0.0273 & 0.0307 & 0.0373 \\
 & (10.2) & (6.7) & (20.2) & (9.6) & (11.7) & (14.3) & (18.4) \\
\addlinespace
PS-alg & 0.0341 & 0.0229 & 0.1687 & 0.0262 & 0.0323 & 0.0820 & 0.1516 \\
 & (14.2) & (7.1) & (54.5) & (10.8) & (15.4) & (35.1) & (51.2) \\
\addlinespace
\bottomrule\end{tabular*}
    \caption{Distribution of assignment correlation for different designs \label{tab:dist_as_cor}}
    \begin{tablenotes}
      \item[] \footnotesize Note: The table presents data from the same simulated distributions of the assignment correlation as is shown in Figure \ref{fig:dist_as_cor}. The corresponding MSE increase from a standard deviation increase relative to its expectation (see equation \ref{eq:rel_var_increase}) is shown in parentheses. $N=50$ and for each sample, the assignment correlations are approximated with 10,000 assignment vectors (including mirrors). 10,000 random samples are drawn.
    \end{tablenotes}
\end{threeparttable}
\end{table}

We now see quite a different picture where more restrictive designs having a relatively larger assignment correlation on average, but where the distribution also has a large right tail. Indeed, for the pair-switching algorithm, the assignment correlation can become quite extreme with values up to eight times as large as the expectation and the percent increase in the MSE from a standard deviation increase relative to its expectation can now be as large as 50 percent for the most extreme designs.

The reason for this extreme outcome is that data contain outliers. In such cases there can be situations where only a small set of very similar assignment vectors can yield small Mahalanobis distances. Because all assignment vectors in such a design are similar to each other, the variability of the MSE becomes huge, and using such a design comes with an inherent risk.

Because the variance of the MSE converges to zero at the rate $N^2$, the risks discussed in this paper are mainly relevant for smaller experiments. In the supporting information, we reproduce the results in Figure \ref{fig:dist_as_cor} and Table \ref{tab:dist_as_cor} for $N=100$. While the assignment correlations are smaller on average compared to $N=50$, there are situations when the assignment correlations take such extreme values that a standard deviation increase in the MSE relative to its expectation can be just as large as for the smaller sample size.

\section{Discussion \label{sec:discussion}}
The theoretical results in Section \ref{sec:theory} says that the least risky design is complete randomization, with more restrictive designs carrying a greater risk of a high MSE. Simulation results indicate that with designs that balances on well-behaved covariates---and where the designs are not too restrictive---this risk is fairly small. However, when covariates include outliers, there is a real chance of having designs that consist of assignment vectors which are very similar to each other, implying that the variance of the MSE is high.

The theoretical results are valid under the veil of ignorance where covariates carry no information about the outcome, with the expected MSE being the same for all designs  (theorem \ref{thm:mean_fix}). In reality, restricted designs such as rerandomization is used when covariates are expected to be informative in explaining the outcome, thereby decreasing the MSE. It is therefore natural to ask what relevance the preceding analysis holds in applied settings.

There are several reasons for why we consider our approach to be relevant. First, before an experiment has taken place, it is in general unknown whether covariates explain the outcome and researchers may be inclined to include covariates in designs even when there are no strong \emph{a priori} reason for their inclusion. In line with, e.g., \cite{efron_forcing_1971} and \cite{Wu1981}, it is therefore relevant to study the risk of different designs with no prior information. In cases when covariates are in fact uninformative, we can view a restricted design based on these covariates as being randomly sampled from behind the veil of ignorance. Our results therefore inform the risk of restricting a design when there is little evidence of covariates being important.

Second, even when covariates can explain the outcome, in many cases the question arise of just \emph{how restricted} a design should be, with no satisfactory answer given in the literature (\citealt{Li2018}) and with modern algorithmic designs, it is possible for designs to include extremely few assignment vectors (see \citealt{Johansson_on_2019} and \citealt{kallus_on_2020}). We consider the assignment correlation to be useful in dealing with the question of how restricted a design should be. For instance, at a certain point, further restricting a design would have negligible impact on the expected MSE, but may increase the assignment correlation. In such a case, there is a potential tradeoff beteween continuously restricting a design, with an incrementally smaller expected MSE, and an increased risk due to all admissible assignment vectors being very similar to each other.

With rerandomization, an algorithmic way of deciding $p_A$ (i.e., just how restricted the design should be) would be to set a threshold value for the largest admissible assignment vector. One can then continuously rerandomize, gradually lowering $p_A$ and calculate the assignment vector. $p_A$ is then chosen based on when the assignment correlation exceeds the threshold value or when the alloted time for rerandomization runs out. In this way, the experimenter can get as low expected MSE as possible, while at the same time feel confident that the resulting design is not too narrow and retains sufficient variation in the assignments. Similar algorithms can be used for various algorithmic designs.

Simulation results indicate that with well-behaved covariates, the tradeoff between lowering expected MSE and an increased assignment correlation is limited and with current computational capacity, the algorithm suggested in the previous paragraph would in most cases terminate when the alloted time is up, rather than due to the assignment correlation being large. However, with increasing computational capacity, and with more clever algorithms for efficiently searching through the space of assignment vectors, this result may change in the near future, making it even more relevant to use a measure such as the assignment correlation to assure sufficient degree of variation is left in the resulting assignment vectors.

If, on the other hand, the assignment correlation increases very quickly when using the algorithm discussed above, it is a sign that the design is being driven by outliers, forcing a certain configuration of treated and control to ensure balance. In such a case, the experimenter may consider trimming or excluding offending covariates.

Finally, the assignment correlation measures the risk of getting a large MSE for a given design. In most cases, it seems likely that the experimenter is risk-averse, meaning that a large assignment correlation should be avoided. A risk-neutral experimenter on the other hand would be indifferent between different assignment correlations whereas if he or she is risk-loving, a large assignment correlation may even be preferable. It is also possible that the experimenter's risk preference depends on whether the study is well-powered or not (see the recent work of \citealt{krieger_improving_2020} for an analysis of the relationship between power of randomization tests and correlation between assignment vectors). A complete analysis of this issue would be well-studied in a decision-theoretic framework, something that is beyond the scope of this paper, but which would be an interesting avenue for future research.

\bibliographystyle{apalike}
\bibliography{library}

\begin{thebibliography}{}

\bibitem[Anscombe, 1948]{yates1948}
Anscombe, F.~J. (1948).
\newblock The validity of comparative experiments (discussion by {Y}ates).
\newblock {\em Journal of the Royal Statistical Society. Series A (General)},
  111(3):181--211.

\bibitem[Efron, 1971]{efron_forcing_1971}
Efron, B. (1971).
\newblock Forcing a sequential experiment to be balanced.
\newblock {\em Biometrika}, 58(3):403--417.

\bibitem[Johansson et~al., 2019]{Johansson_on_2019}
Johansson, P., Rubin, D., and Schultzberg, M. (2019).
\newblock {On optimal re-randomization designs}.
\newblock Working Paper, Department of Statistics, Uppsala University, 2019:3.
  \textit{Forthcoming in Journal of the Royal Statistical Society. Series B
  (Statistical Methodology)}.

\bibitem[Johansson and Schultzberg, 2020]{Johansson2018a}
Johansson, P. and Schultzberg, M. (2020).
\newblock {Rerandomization strategies for balancing covariates using
  pre-experimental longitudinal data}.
\newblock {\em Journal of Computational and Graphical Statistics},
  29(4):798--813.

\bibitem[Kallus, 2018]{Kallus2018}
Kallus, N. (2018).
\newblock {Optimal a priori balance in the design of controlled experiments}.
\newblock {\em Journal of the Royal Statistical Society. Series B (Statistical
  Methodology)}, 80(1):85--112.

\bibitem[Kallus, 2020]{kallus_on_2020}
Kallus, N. (2020).
\newblock On the optimality of randomization in experimental design: How to
  randomize for minimax variance and design-based inference.
\newblock {\em arXiv e-prints}, page arXiv:2005.03151.
\newblock \textit{Forthcoming in Journal of the Royal Statistical Society.
  Series B (Statistical Methodology)}.

\bibitem[{Kapelner} et~al., 2019]{kapelner2019}
{Kapelner}, A., {Krieger}, A.~M., {Sklar}, M., and {Azriel}, D. (2019).
\newblock {Optimal rerandomization via a criterion that provides insurance
  against failed experiments}.
\newblock {\em arXiv e-prints}, page arXiv:1905.03337.

\bibitem[Kapelner et~al., 2020]{kapelner_harmonizing_2020}
Kapelner, A., Krieger, A.~M., Sklar, M., Shalit, U., and Azriel, D. (2020).
\newblock Harmonizing {optimized} {designs} with {classic} {randomization} in
  {experiments}.
\newblock {\em The American Statistician. Forthcoming}.

\bibitem[Krieger et~al., 2019]{Krieger2019}
Krieger, A.~M., Azriel, D., and Kapelner, A. (2019).
\newblock {Nearly random designs with greatly improved balance}.
\newblock {\em Biometrika}, 106(3):695--701.

\bibitem[Krieger et~al., 2020]{krieger_improving_2020}
Krieger, A.~M., Azriel, D., Sklar, M., and Kapelner, A. (2020).
\newblock Improving the power of the randomization test.
\newblock {\em arXiv e-prints}, page arXiv:2008.05980.

\bibitem[Lauretto et~al., 2017]{Lauretto2017}
Lauretto, M.~S., Stern, R.~B., Morgan, K.~L., Clark, M.~H., and Stern, J.~M.
  (2017).
\newblock {Haphazard intentional allocation and rerandomization to improve
  covariate balance in experiments}.
\newblock {\em AIP Conference Proceedings}, 1853.

\bibitem[Li et~al., 2018]{Li2018}
Li, X., Ding, P., and Rubin, D.~B. (2018).
\newblock {Asymptotic theory of rerandomization in treatment-control
  experiments}.
\newblock {\em Proceedings of the National Academy of Sciences}, 115(37):9157
  -- 9162.

\bibitem[Morgan and Rubin, 2012]{Morgan2012}
Morgan, K.~L. and Rubin, D.~B. (2012).
\newblock {Rerandomization to improve covariate balance in experiments}.
\newblock {\em Annals of Statistics}, 40(2):1263--1282.

\bibitem[Rosenberger and Lachin, 2015]{rosenberger_randomization_2015}
Rosenberger, W.~F. and Lachin, J.~M. (2015).
\newblock {\em Randomization in {clinical} {trials} : {theory} and {practice}}.
\newblock John Wiley \& Sons, Incorporated, New York, United States.

\bibitem[Schultzberg and Johansson, 2019]{Schultzberg2019c}
Schultzberg, M. and Johansson, P. (2019).
\newblock {Rerandomization: A complement or substitute for stratification in
  randomized experiments?}
\newblock Working Paper, Department of Statistics, Uppsala University, 2019:4.

\bibitem[Wu, 1981]{Wu1981}
Wu, C.-F. (1981).
\newblock {On the robustness and efficiency of some randomized designs}.
\newblock {\em The Annals of Statistics}, 9(6):1168--1177.

\bibitem[Youden, 1972]{youden_randomization_1972}
Youden, W.~J. (1972).
\newblock Randomization and {experimentation}.
\newblock {\em Technometrics}, 14(1):13--22.

\end{thebibliography}

\appendix

\clearpage

\begin{center}
{\large\bf SUPPLEMENTARY MATERIAL}
\end{center}

\setcounter{equation}{0}
\renewcommand{\theequation}{A\arabic{equation}}

\newtheorem{theoremA}{Theorem}
\renewcommand{\thetheoremA}{A\arabic{theoremA}}

\addcontentsline{toc}{section}{Appendices}
\renewcommand{\thesubsection}{\Alph{subsection}}

\setcounter{table}{0}
\renewcommand{\thetable}{A\arabic{table}}

\setcounter{figure}{0}
\renewcommand{\thefigure}{A\arabic{figure}}

\subsection*{Proof of theorem \ref{thm:mean_fix}}
The average MSE over all designs in $\tilde{\mathcal{K}}_H$ can be written as
    \begin{equation}
        \frac{1}{m}\sum_{k=1}^{m}\mse(\hat\tau_{\{\mathcal{\tilde W}_H^k\}}) =
        \frac{1}{m}\sum_{k=1}^{m}\left( \frac{1}{H}\sum_{\mathbf{w}_j\in \tilde{\mathcal{W}}_H^k} (\hat\tau_j - \tau)^2 \right). \label{eq:avg_var}
    \end{equation}
    Let $I_{jk}:=1[\mathbf{w}_j\in \tilde{\mathcal{W}}_H^k]$ be an indicator for whether assignment vector $\mathbf{w}_j$ occurs in design $\tilde{\mathcal{W}}_H^k$. We can rewrite equation \eqref{eq:avg_var} as
    \begin{equation}
        \frac{1}{m}\sum_{k=1}^{m}\mse(\hat\tau_{\{\mathcal{\tilde W}_H^k\}}) =
        \frac{1}{m}\sum_{k=1}^{m}\sum_{j=1}^{N_A}\frac{I_{jk}}{H} (\hat\tau_j - \tau)^2
    \end{equation}
    Together with equation \eqref{eq:var_cr} in the paper, we see that this expression equals $\sigma^2_{CR}$ if
    \begin{equation}
        \sum_{k=1}^{m}\sum_{j=1}^{N_A}\frac{I_{jk}}{H}(\hat\tau_j - \tau)^2 = \sum_{j=1}^{N_A}\frac{m}{N_A}(\hat\tau_j - \tau)^2.
    \end{equation}
    By condition \ref{cond:cond1} we know that $\sum_{k=1}^m I_{jk}=c$, i.e. that the total number of times an assignment vector occurs in the designs in $\tilde{\mathcal{K}}_H$ is the same for every assignment vector. We have
    \begin{equation}
        \sum_{j=1}^{N_A} c(\hat\tau_j - \tau)^2 = \sum_{j=1}^{N_A} \frac{m H}{N_A}(\hat\tau_j - \tau)^2. \label{eq:exp_var_cr}
    \end{equation}
    Because there are $m$ designs in $\tilde{\mathcal{K}}_H$ with $H$ assignment vectors in each design, there are $mH$ total number of assignments vectors over all designs in $\tilde{\mathcal{K}}_H$. There are $N_A$ different assignment vectors and so the number of times each assignment vector occurs over all designs is $c=\frac{mH}{N_A}$ which completes the proof.

\subsection*{Proof of theorem \ref{thm:max_fix}}
Let $r_j:=(\hat\tau_j - \tau)^2$. Because we always include mirror assignments in any design under consideration, it is the case that $r_j=r_{N_A+1-j}$ if estimates are ordered lexicographically. Without loss of generality, we can therefore study only the first half of the lexicographic ordering. The order statistics for these estimates are: $o_{(1)},\ldots,o_{(N_A/2)}$ (i.e., it is the ordering of $r_1,\ldots,r_{N_A/2}$). It is the case that
    \begin{equation}
        \max \left\{\mse(\hat\tau_{\{\mathcal{W}_H^k\}}) : k = 1,\ldots,\binom{N_A/2}{H/2}\right\} = \frac{1}{H/2} \sum_{j=1}^{H/2} o_{(N_A/2+1-j)}.
    \end{equation}
    We are interested in the difference in maximum MSE between $\mathcal{K}_H$ and $\mathcal{K}_{H'}$:
    \begin{equation}
        \frac{1}{H/2} \sum_{j=1}^{H/2} o_{(N_A/2+1-j)} - \frac{1}{H'/2} \sum_{j=1}^{H'/2} o_{(N_A/2+1-j)}.
    \end{equation}
    From the definition of order statistics, it must be the case that
    \begin{equation}
        o_{(N_A/2+1-v)} \leq \frac{1}{H/2} \sum_{j=1}^{H/2} o_{(N_A/2+1-j)}, \quad \forall v = H+1\ldots H'.
    \end{equation}
    If $o_{(N_A/2+1-H')} < \frac{1}{H/2} \sum_{j=1}^{H/2} o_{(N_A/2+1-j)}$, then the inequality in the theorem is strict. The exact same argument can be used to show that the minimum MSE is increasing in $H$.
\subsection*{Proof of theorem \ref{thm:varvar_h}}
Let $m=\binom{N_A/2}{H/2}$, we can expand equation \eqref{eq:def_var_mse} in the paper to
\begin{multline}
    \var\left(\mse(\hat\tau_{\{\mathcal{W}_H\}}):\mathcal{W}_H\in \mathcal{K}_H\right)=\\\frac{1}{m} \sum_{k=1}^{m} \left(
    \left(\sigma^2_{CR}\right)^2
    -\frac{2}{H}\sum_{\mathbf{w}_j \in \mathcal{W}_{H}^k}(\hat\tau_j-\tau)^2\sigma^2_{CR}
    +\frac{1}{H^2} \sum_{\mathbf{w}_j \in \mathcal{W}_{H}^k}\sum_{\mathbf{w}_{j'} \in \mathcal{W}_{H}^k} (\hat\tau_j-\tau)^2(\hat\tau_{j'}-\tau)^2 \right).\label{eq:varvar_kh} 
\end{multline}
By condition \ref{cond:cond1}, we know that each assignment vector will occur in $\frac{mH}{N_A}$ designs. Using this information together with the mirror-property, $\hat\tau_j - \tau = \hat\tau_{N_A+1-j} + \tau$, and once again using the notation $r_j=(\hat\tau_j - \tau)^2$, we get
\begin{multline}
    \var\left(\mse(\hat\tau_{\{\mathcal{W}_H\}}):\mathcal{W}_H\in \mathcal{K}_H\right)=\left(\sigma^2_{CR}\right)^2 - 2\frac{\sigma^2_{CR}}{N_A} \sum_{j=1}^{N_A} r_j + \\
     \frac{2}{HN_A}\sum_{j=1}^{N_A} r_j^2 + \frac{H-2}{HN_A(N_A-2)} \sum_{j=1}^{N_A}\sum_{j' \neq j, -j}^{N_A} r_jr_{j'},\label{eq:varvarkh}
\end{multline}
where $\sum_{j' \neq j, -j}^{N_A}$ is shorthand notation for going through all assignment vectors except for $j$ and its mirror, $-j$. By equation \eqref{eq:var_cr} in the paper, we know that 

\begin{equation}
    \left(\sigma^2_{CR}\right)^2 - 2\frac{\sigma^2_{CR}}{N_A} \sum_{j=1}^{N_A} r_j = - \frac{1}{N_A^2} \left(\sum_{j=1}^{N_A} r_j\right)^2 = -\frac{2}{N_A^2}\sum_{j=1}^{N_A}r_j^2 - \frac{1}{N_A^2} \sum_{i=1}^{N_A} \sum_{j' \neq j, -j}^{N_A} r_jr_{j'}, \label{eq:sigma2cr}
\end{equation}
where we have used the fact that $r_jr_{N_A+1-j}=r_j^2$.

Let $\bar{p}:=\frac{1}{N_A}\sum_{j=1}^{N_A}r_j^2$ and $\bar{q}:=\frac{1}{N_A(N_A-2)}\sum_{j=1}^{N_A}  \sum_{j'\neq j, -j}^{N_A}  r_jr_{j'}$ be the average value of $r_j^2$ and $r_jr_{j'}$, respectively. Combining equations \eqref{eq:varvarkh} and \eqref{eq:sigma2cr}, we get

\begin{equation}
    \var\left(\mse(\hat\tau_{\{\mathcal{W}_H\}}):\mathcal{W}_H\in \mathcal{K}_H\right) = 2\frac{ N_A - H}{HN_A}(\bar{p} -\bar{q}). \label{eq:varvar_h}
\end{equation}
Because $\bar{p} -\bar{q} \geq 0$, the proof is completed. Note that if for some $j\neq j'$ it is the case that $\hat\tau_j\neq \hat\tau_{j'}$, then $\bar{p} -\bar{q} > 0$ and so the variance of the MSE is strictly decreasing in $H$.

\subsection*{Proof of theorem \ref{thm:varvar_u}}
From equation \eqref{eq:varvar_kh}, we get
\begin{multline}
    \var\left(\mse(\hat\tau_{\{\mathcal{W}_H\}}):\mathcal{W}_H\in \mathcal{\tilde K}_H\right)=\\\frac{1}{m} \sum_{k=1}^{m} \left(
    \left(\sigma^2_{CR}\right)^2
    -\frac{2}{H}\sum_{\mathbf{w}_a \in \mathcal{\tilde{W}}_{H}^k}(\hat\tau_a-\tau)^2\sigma^2_{CR}
    +\frac{1}{H^2} \sum_{\mathbf{w}_a \in \mathcal{\tilde{W}}_{H}^k}\sum_{\mathbf{w}_{b}  \in \mathcal{\tilde{W}}_{H}^k} (\hat\tau_a-\tau)^2(\hat\tau_{b}-\tau)^2 \right),
\end{multline}
where $m=|\mathcal{\tilde K}_H|$ is the number of designs in the set $\mathcal{\tilde K}_H$. Note that we switch to using $a$ and $b$ as indexes instead of $j$ and $j'$ because we will use the $j$-index to sum over observations below. For each design there are $H^2$ pairs and so the proportion of all pairs of assignment vectors which has a given $u$ over all the designs in the set $\mathcal{\tilde K}_H$ is
\begin{equation}
    v_u(\mathcal{\tilde{K}}_H) := \frac{1}{mH^2}\sum_{k=1}^{m}\sum_{\mathbf{w}_a \in \mathcal{\tilde{W}}_{H}^k}\sum_{\mathbf{w}_{b}  \in \mathcal{\tilde{W}}_{H}^k}  1[U_{a,b}=u]
\end{equation}
Let $n_u:=N_A\binom{N/2}{u}\binom{N/2}{N/2-u}$ be the total number of pairs with a given $u$ out of the $N_A^2$ possible pairs.

Once again, we use the notation $r_a = (\hat\tau_a-\tau)^2$. By condition \ref{cond:cond1}, we know that each assignment vector will occur in $\frac{H}{N_A}m$ designs, and by condition \ref{cond:cond2} we know that two assignment vectors will occur together in a design an equal number of times for all pairs of assignment vectors with the same pairwise uniqueness. Together with equation \eqref{eq:sigma2cr}, we get
\begin{equation}
    \var\left(\mse(\hat\tau_{\{\mathcal{W}_H\}}):\mathcal{W}_H\in \mathcal{\tilde K}_H\right)= \sum_{u=0}^{N/2} \frac{v_u(\mathcal{\tilde{K}}_H)}{n_u} \sum_{a=1}^{N_A}  \sum_{b=1}^{N_A} 1[U_{a,b}=u] r_ar_{b} -\frac{1}{N_A^2} \sum_{a=1}^{N_A}  \sum_{b=1}^{N_A}  r_ar_{b}.
\end{equation}
Let $\bar{q}_u:=\frac{1}{n_u}\sum_{a=1}^{N_A}  \sum_{b=1}^{N_A}  1[U_{a,b}=u] r_ar_{b}$. We get
\begin{equation}
    \var\left(\mse(\hat\tau_{\{\mathcal{W}_H\}}):\mathcal{W}_H\in \mathcal{\tilde K}_H\right)=
    \sum_{u=0}^{N/2}\left(v_u(\mathcal{\tilde{ K}}_H) - \frac{n_u}{N_A^2}\right)\bar{q}_u \label{eq:rel_share_mse}
\end{equation}
The following lemma will be useful:
\begin{lemma}\label{lem:bar_qu}
     $\bar{q}_u$ is a function of $u$ for which it is the case that
     \begin{enumerate}[i]
          \item \label{lem:part_a} the variance of the MSE is independent of any terms in $\bar{q}_u$ that are either constant terms or linear in $u$. 
          \item \label{lem:part_b} if $\bar{q}_u = k_0 + k_1u + k_2u^2$, then the variance of the MSE can be written as 
          \begin{equation}
              \var\left(\mse(\hat\tau_{\{\mathcal{W}_H\}}):\mathcal{W}_H\in \mathcal{\tilde K}_H\right)= k_2\sum_{u=0}^{N/2} \left(v_u(\mathcal{\tilde{ K}}_H) - \frac{n_u}{N_A^2}\right)(u-N/4)^2. 
          \end{equation}
    \end{enumerate} 
\end{lemma}
\begin{proof}
    To see this, let $g_u := v_u(\mathcal{\tilde{ K}}_H)-\frac{n_u}{N_A^2}$ and note that $\sum_{u=0}^{N/2}g_u=0$, and so $\sum_{u=0}^{N/2}g_u\bar{q}_u$ will not depend on $k_0$. Furthermore, $g_u$ is completely symmetrical around $u=N/4$ which means that $\sum_{u=0}^{N/2}g_u\bar{q}_u$ does not depend on $k_1u$. Finally, for the exact same reason, because $(u-N/4)^2$ is a second degree polynomial, it is the case that $\sum_{u=0}^{N/2}g_uu^2=\sum_{u=0}^{N/2}g_u(u-N/4)^2$.
\end{proof}

To find $\bar{q}_u$, note that we can write $r_a$ as
\begin{equation}
    r_a = (\hat\tau_a - \tau)^2 = \left(\frac{1}{(N/2)}\left(\sum_{i \in \mathcal{T}}  Y_i(1)-\sum_{i \in \mathcal{C}} Y_i(0)\right) - \tau\right)^2,
\end{equation}
where $\mathcal{T} \subset \{1,\ldots,N\}$ is the set of indexes for the treated units and $\mathcal{C}=\{1,\ldots,N\} \setminus \mathcal{T} $ the set of indexes for the control units. Let $\alpha_{ix}=Y_i(1)$ if unit $i$ is treated in assignment vector $x$ and $\alpha_{ix}=-Y_i(0)$ if unit $i$ is control in assignment vector $x$. We have
\begin{align}
    r_ar_b = 
    &\tau^4 - \frac{4\tau^3}{N} \left(\sum_{i=1}^N\alpha_{ia} + \sum_{i=1}^N\alpha_{ib}\right) + \nonumber \\ &
    \frac{4\tau^2}{N^2} \left(
        \sum_{i=1}^N\alpha_{ia}^2
        + \sum_{i=1}^N\alpha_{ib}^2
        + 4\sum_{i=1}^N\sum_{j=1}^N\alpha_{ia}\alpha_{jb}
        + 2\sum_{i=1}^{N-1}\sum_{j=i+1}^N \alpha_{ia}\alpha_{ja} 
        + 2\sum_{i=1}^{N-1}\sum_{j=i+1}^N \alpha_{ib}\alpha_{jb} \right) -\nonumber \\
    &
    \frac{16\tau}{N^3}
    \left(
        \sum_{i=1}^N\sum_{j=1}^N \alpha_{ia}^2\alpha_{jb} +
        \sum_{i=1}^N\sum_{j=1}^N \alpha_{jb}^2\alpha_{ia} +
        2\sum_{i=1}^{N-1}\sum_{j=1}^N\sum_{k=i+1}^N \alpha_{ia}\alpha_{jb}\alpha_{ka} +
        2\sum_{i=1}^{N-1}\sum_{j=1}^N\sum_{k=i+1}^N \alpha_{ib}\alpha_{ja}\alpha_{kb}
        \right)+ \nonumber\\
    &\frac{16}{N^4} \left( 
        \sum_{i=1}^N\sum_{j=1}^N \alpha_{ia}^2\alpha_{jb}^2 +
        2\sum_{i=1}^N\sum_{j=1}^{N-1}\sum_{k=j+1}^N \alpha_{ia}^2\alpha_{jb}\alpha_{kb} +
        2\sum_{i=1}^N\sum_{j=1}^{N-1}\sum_{k=j+1}^N \alpha_{ib}^2\alpha_{ja}\alpha_{ka} + \right. \nonumber \\
        & \left. 4\sum_{i=1}^{N-1}\sum_{j=1}^{N-1}\sum_{k=i+1}^N\sum_{l=j+1}^N\alpha_{ia}\alpha_{jb}\alpha_{ka}\alpha_{lb}
    \right). \label{eq:sum_qbar_u}
\end{align}
When we sum this expression over all assignment vectors with a given $u$, by conditions \ref{cond:cond1} and \ref{cond:cond2}, only the last sum will depend on $u$. The reason is that in all other cases, the sums never involve an expression for how often two units in assignment $a$ occur together with two units in assignment $b$. Hence, by lemma \ref{lem:bar_qu}\ref{lem:part_a}), we get
\begin{multline}
    \var\left(\mse(\hat\tau_{\{\mathcal{W}_H\}}):\mathcal{W}_H\in \mathcal{\tilde K}_H\right)= \\
    \frac{64}{N^4}\sum_{u=0}^{N/2}\left(v_u(\mathcal{\tilde{ K}}_H) - \frac{n_u}{N_A^2}\right)\frac{1}{n_u}\sum_{a=1}^{N_A}  \sum_{b=1}^{N_A}  1[U_{a,b}=u] \sum_{i=1}^{N-1}\sum_{j=1}^{N-1}\sum_{k=i+1}^N\sum_{l=j+1}^N\alpha_{ia}\alpha_{jb}\alpha_{ka}\alpha_{lb}. \label{eq:rel_share_mse2}
\end{multline}

The quadruple sum can be written as
\begin{multline}
    \sum_{i=1}^{N-1}\sum_{j=1}^{N-1}\sum_{k=i+1}^N\sum_{l=j+1}^N\alpha_{ia}\alpha_{jb}\alpha_{ka}\alpha_{lb} =
    \sum_{i=1}^{N-1}\sum_{j=i+1}^{N} \alpha_{ia}\alpha_{ib}\alpha_{ja}\alpha_{jb} + 
    \sum_{i=1}^{N-2}\sum_{j=i+1}^{N-1}\sum_{k=j+1}^{N}
    \left( 
        \alpha_{ia}\alpha_{ib}\alpha_{ja}\alpha_{kb} + \right.\\\left.
        \alpha_{ia}\alpha_{ib}\alpha_{jb}\alpha_{ka} +
        \alpha_{ia}\alpha_{ja}\alpha_{jb}\alpha_{kb} +
        \alpha_{ia}\alpha_{jb}\alpha_{ka}\alpha_{kb} +
        \alpha_{ib}\alpha_{ja}\alpha_{jb}\alpha_{ka} + 
        \alpha_{ib}\alpha_{ja}\alpha_{ka}\alpha_{kb}
        \right) + \\
        \sum_{i=1}^{N-3}\sum_{j=i+1}^{N-2}\sum_{k=j+1}^{N-1}\sum_{l=k+1}^N
        \left( 
        \alpha_{ia}\alpha_{ja}\alpha_{kb}\alpha_{lb} +
        \alpha_{ia}\alpha_{jb}\alpha_{ka}\alpha_{lb} +
        \alpha_{ia}\alpha_{jb}\alpha_{kb}\alpha_{la} +
        \alpha_{ib}\alpha_{ja}\alpha_{ka}\alpha_{lb} +
        \right. \\
        \left.
        \alpha_{ib}\alpha_{ja}\alpha_{kb}\alpha_{la} +
        \alpha_{ib}\alpha_{jb}\alpha_{ka}\alpha_{la}
    \right)\label{eq:aiajakal}
\end{multline}

Define the $E_u(\cdot)$ as the expectations operator over all pairwise assignments with a given $u$. I.e., 
\begin{equation}
    E_u\left(\sum_{i=1}^{N-1}\sum_{j=1}^{N-1}\sum_{k=i+1}^N\sum_{l=j+1}^N\alpha_{ia}\alpha_{jb}\alpha_{ka}\alpha_{lb}\right) = \frac{1}{n_u}\sum_{a=1}^{N_A}  \sum_{b=1}^{N_A}  1[U_{a,b}=u] \sum_{i=1}^{N-1}\sum_{j=1}^{N-1}\sum_{k=i+1}^N\sum_{l=j+1}^N\alpha_{ia}\alpha_{jb}\alpha_{ka}\alpha_{lb}
\end{equation}
By conditions \ref{cond:cond1} and \ref{cond:cond2}, the $a$- and $b$-indexes are interchangeable and so we have
\begin{multline}
    E_u\left( \sum_{i=1}^{N-1}\sum_{j=1}^{N-1}\sum_{k=i+1}^N\sum_{l=j+1}^N\alpha_{ia}\alpha_{jb}\alpha_{ka}\alpha_{lb} \right) = E_u\left(\sum_{i=1}^{N-1}\sum_{j=i+1}^{N} \alpha_{ia}\alpha_{ib}\alpha_{ja}\alpha_{jb}\right) + \\
    2E_u\left(\sum_{i=1}^{N-2}\sum_{j=i+1}^{N-1}\sum_{k=j+1}^{N}\alpha_{ia}\alpha_{ib}\alpha_{ja}\alpha_{kb}\right) +
    2E_u\left(\sum_{i=1}^{N-2}\sum_{j=i+1}^{N-1}\sum_{k=j+1}^{N}\alpha_{ia}\alpha_{ja}\alpha_{jb}\alpha_{kb}\right) + \\
    2E_u\left(\sum_{i=1}^{N-2}\sum_{j=i+1}^{N-1}\sum_{k=j+1}^{N}\alpha_{ia}\alpha_{jb}\alpha_{ka}\alpha_{kb}\right) + 2E_u\left(\sum_{i=1}^{N-3}\sum_{j=i+1}^{N-2}\sum_{k=j+1}^{N-1}\sum_{l=k+1}^N\alpha_{ia}\alpha_{ja}\alpha_{kb}\alpha_{lb}\right) + \\
    2E_u\left(\sum_{i=1}^{N-3}\sum_{j=i+1}^{N-2}\sum_{k=j+1}^{N-1}\sum_{l=k+1}^N\alpha_{ia}\alpha_{jb}\alpha_{ka}\alpha_{lb}\right) +
    2E_u\left(\sum_{i=1}^{N-3}\sum_{j=i+1}^{N-2}\sum_{k=j+1}^{N-1}\sum_{l=k+1}^N\alpha_{ia}\alpha_{jb}\alpha_{kb}\alpha_{la}\right). \label{eq:exp_iajbkalb}
\end{multline}

It is the case that 
\begin{multline}
     E_u\left(\sum_{i=1}^{N-1}\sum_{j=i+1}^{N} \alpha_{ia}\alpha_{ib}\alpha_{ja}\alpha_{jb}\right) = \sum_{i=1}^{N-1}\sum_{j=i+1}^{N}p_{1111}' Y_i(1)^2Y_j(1)^2 -
    (p_{1110}' + p_{1101}') Y_i(1)^2Y_j(1)Y_j(0) - \\
    (p_{1011}' + p_{0111}') Y_i(1)Y_i(0)Y_j(1)^2+ 
    (p_{1010}' + p_{0101}' + p_{1001}' + p_{0110}') Y_i(1)Y_i(0)Y_j(1)Y_j(0)- \\
    (p_{0001}' + p_{0010}') Y_i(0)^2Y_j(1)Y_j(0)- (p_{1000}' + 
    p_{0100}') Y_i(1)Y_i(0)Y_j(0)^2+ 
    p_{0011'} Y_i(1)^2Y_j(1)^2 + \\
     p_{1100}' Y_i(1)^2Y_j(0)^2 + p_{0000}' Y_i(0)^2Y_j(0)^2,
\end{multline}
where $p_{x_1x_2x_3x_4}'=\Pr(w_{a}^i=x_1\cap w_{b}^i=x_2 \cap w_{a}^j = x_3 \cap w_{b}^j = x_4)$ and $w_{a}^i$ being the treatment status of unit $i$ in assignment vector $a$. More compactly, it can be written as 
\begin{multline}
     E_u\left(\sum_{i=1}^{N-1}\sum_{j=i+1}^{N} \alpha_{ia}\alpha_{ib}\alpha_{ja}\alpha_{jb}\right) = \sum_{i=1}^{N-1}\sum_{j=i+1}^{N}\sum_{x_1=0}^1\sum_{x_2=0}^1\sum_{x_3=0}^1\sum_{x_4=0}^1 \pi p_{x_1x_2x_3x_4}' Y_i(x_1)Y_i(x_2)Y_j(x_3)Y_j(x_4), \label{eq:pr_double_sum}
\end{multline}
where $\pi$ equals 1 if $x_1+x_2+x_3+x_4$ is even and 0 otherwise.

Similarly, we have
\begin{multline}
     E_u\left(\sum_{i=1}^{N-2}\sum_{j=i+1}^{N-1}\sum_{k=j+1}^{N}\alpha_{ia}\alpha_{ib}\alpha_{ja}\alpha_{kb}\right) = \\ \sum_{i=1}^{N-2}\sum_{j=i+1}^{N-1}\sum_{k=j+1}^{N}\sum_{x_1=0}^1\sum_{x_2=0}^1\sum_{x_3=0}^1\sum_{x_4=0}^1 \pi p_{x_1x_2x_3x_4}'' Y_i(x_1)Y_i(x_2)Y_j(x_3)Y_k(x_4), \label{eq:pr_triple_sum}
\end{multline}
where $p_{x_1x_2x_3x_4}''=\Pr(w_{a}^i=x_1\cap w_{b}^i=x_2 \cap w_{a}^j = x_3 \cap w_{b}^k = x_4)$. To get the other two terms of the triple sum, the indexes just have to be changed. Finally,
\begin{multline}
     E_u\left(\sum_{i=1}^{N-3}\sum_{j=i+1}^{N-2}\sum_{k=j+1}^{N-1}\sum_{l=k+1}^N\alpha_{ia}\alpha_{ja}\alpha_{kb}\alpha_{lb}\right) = \\ \sum_{i=1}^{N-3}\sum_{j=i+1}^{N-2}\sum_{k=j+1}^{N-1}\sum_{l=k+1}^N\sum_{x_1=0}^1\sum_{x_2=0}^1\sum_{x_3=0}^1\sum_{x_4=0}^1 \pi p_{x_1x_2x_3x_4}''' Y_i(x_1)Y_j(x_2)Y_k(x_3)Y_l(x_4), \label{eq:pr_quadruple_sum}
\end{multline}
where $p_{x_1x_2x_3x_4}'''=\Pr(w_{a}^i=x_1\cap w_{a}^j=x_2 \cap w_{b}^k = x_3 \cap w_{b}^l = x_4)$. Once again, to get the other two terms of the quadruple sum, the indexes just have to be changed. The question now is to find all probabilities $p_{x_1x_2x_3x_4}'$, $p_{x_1x_2x_3x_4}''$ and $p_{x_1x_2x_3x_4}'''$.

We begin with $p_{x_1x_2x_3x_4}'$ and fix the treatment status of units $i,j$ in both $\mathbf{w}_a$ and $\mathbf{w}_b$. Let $c_1' \in\{0,1,2\}$ be the number of units $i,j$ treated in both $\mathbf{w}_a$ and $\mathbf{w}_b$, $c_2' \in\{0,1,2\}$ is the number of units $i,j$ treated in $\mathbf{w}_a$ but not in $\mathbf{w}_b$ and $c_3' \in\{0,1,2\}$ is the number of units $i,j$ not treated in $\mathbf{w}_a$ but treated in $\mathbf{w}_b$. The number of ways that the remaining $N-2$ other units can be assigned to treatment in $\mathbf{w}_a$ and $\mathbf{w}_b$ is $\binom{N-2}{N/2-c_1'-c_2'}\binom{N/2-c_1'-c_2'}{N/2-U_{a,b}-c_1'}\binom{N-2-N/2+c_1'+c_2'}{U_{a,b}-c_3'}$, where the first binomial coefficient is the number of ways all units but $i$ and $j$ can be selected into treatment in assignment vector $\mathbf{w}_a$, the second term how many of these that are also selected into treatment in assignment vector $\mathbf{w}_b$ and the third the number of ways the rest of the units in $\mathbf{w}_b$ can be selected into treatment. For two random assignment vectors $\mathbf{w}_a$ and $\mathbf{w}_b$, we have
\begin{equation}
    \Pr(c_1',c_2',c_3'|U_{a,b}=u) = \frac{\binom{N-2}{N/2-c_1'-c_2'}\binom{N/2-c_1'-c_2'}{N/2-u-c_1'}\binom{N-2-N/2+c_1'+c_2'}{u-c_3'}}{\binom{N}{N/2}\binom{N/2}{N/2-u}\binom{N/2}{u}},
    \label{eq:p2}
\end{equation}
where the denominator is the total number of ways to select treatment in two assignment vectors with a given $u$.

To find $p_{x_1x_2x_3x_4}'$, note that each combination of $x_1,x_2,x_3,x_4$, by definition, maps to one of the $2^4=16$ permutations of treatment statuses in $i,j$ for two random assignment vectors, each being equally likely. Each permutation in turn implies a given $\Pr(c_1',c_2',c_3'|U_{a,b}=u)$. So, to find, for instance, $p_{1011}'$, we use the fact that this implies $c_1'=1,c_2'=0,c_3'=1$. Using this methodology, we get
\begin{align*}
    &p_{1111}' = p_{0000}' = \frac{1}{\Perm{N}{2}}\left(\frac{N^2}{4}-\frac{N}{2}-Nu + u^2+u\right), \\
    &p_{1110}' = p_{1101}' = p_{1011}' = p_{0111}' = p_{0001}' = p_{0010}' = p_{0100}' = p_{1000}' =  \frac{1}{\Perm{N}{2}}\left(\frac{N}{2}u-u^2\right), \\
    &p_{1010}' = p_{0101}' = \frac{1}{\Perm{N}{2}}\left(u^2-u\right), \\
    &p_{1001}' = p_{0110}' = \frac{1}{\Perm{N}{2}}u^2, \\
    &p_{1100}'=p_{0011}' =\frac{1}{\Perm{N}{2}}\left(\frac{N^2}{4} - Nu + u^2\right),
\end{align*}
where $\Perm{N}{2}$ is the number of arrangments, $N!/(N-2)!=N(N-1)$.

Let $c_1'' \in\{0,1,2,3\}$ be the number of units $i,j,k$ treated in both $\mathbf{w}_a$ and $\mathbf{w}_b$, $c_2'' \in\{0,1,2,3\}$ is the number of units $i,j,k$ treated in $\mathbf{w}_a$ but not in $\mathbf{w}_b$ and $c_3'' \in\{0,1,2,3\}$ is the number of units $i,j,k$ not treated in $\mathbf{w}_a$ but treated in $\mathbf{w}_b$. Analogous to the reasoning above, we now have
\begin{equation}
    \Pr(c_1'',c_2'',c_3''|U_{a,b}=u) = \frac{\binom{N-3}{N/2-c_1''-c_2''}\binom{N/2-c_1''-c_2''}{N/2-u-c_1''}\binom{N-3-N/2+c_1''+c_2''}{u-c_3''}}{\binom{N}{N/2}\binom{N/2}{N/2-u}\binom{N/2}{u}}.
    \label{eq:p3}
\end{equation}
To find $p_{x_1x_2x_3x_4}''$, note that each combination of $x_1,x_2,x_3,x_4$, now maps to four of the $2^6=64$ permutations of treatment statuses in $i,j,k$ for two random assignment vectors, each being equally likely. For instance, to get $p_{1011}''$, aside from $w_{a}^i=1,w_{a}^j=0,w_{b}^i=1,w_{b}^k=1$, there are now four possible permutations of $w_{a}^k,w_{b}^j$, each implying a different $c_1''$, $c_2''$, $c_3''$. To get $p_{1011}''$, simply add these four implied probabilities in equation \eqref{eq:p3} together. Doing so for all $p_{x_1x_2x_3x_4}''$, we get
\begin{align*}
    &p_{1111}'' = p_{0000}'' = \frac{1}{\Perm{N}{3}}\left(\frac{N^3}{8} - \frac{N^2}{4}u - \frac{3N^2}{4} + 2Nu + N - u^2 - 2u\right), \\
    &p_{1110}'' = p_{1101}'' = p_{0001}''=p_{0010}'' = \frac{1}{\Perm{N}{3}}\left(\frac{N^2}{4}u -Nu+  u^2\right), \\
    &p_{0111}'' = p_{1011}'' = p_{1000}'' = p_{0100}'' = \frac{1}{\Perm{N}{3}}\left(\frac{N^3}{8}-\frac{N^2}{4}u-\frac{N^2}{4} + u^2\right), \\
    &p_{1001}'' = p_{0110}'' = \frac{1}{\Perm{N}{3}}\left(\frac{N^2}{4}u - u^2\right), \\
    &p_{1100}'' = p_{0011}'' =  \frac{1}{\Perm{N}{3}}\left(\frac{N^2}{4}u - Nu - u^2 + 2u\right),\\
    &p_{1010}'' = p_{0101}'' = \frac{1}{\Perm{N}{3}}\left(\frac{N^3}{8} - \frac{N^2}{4}u - \frac{N^2}{4} + Nu - u^2\right). \\
\end{align*}

Finally, let $c_1''' \in\{0,1,2,3,4\}$ be the number of units $i,j,k,l$ treated in both $\mathbf{w}_a$ and $\mathbf{w}_b$, $c_2''' \in\{0,1,2,3,4\}$ the number of units  $i,j,k,l$ treated in $\mathbf{w}_a$ but not in $\mathbf{w}_b$ and $c_3'' \in\{0,1,2,3\}$  the number of units $i,j,k,l$ not treated in $\mathbf{w}_a$ but treated in $\mathbf{w}_b$. We now have
\begin{equation}
    \Pr(c_1''',c_2''',c_3'''|u=u) = \frac{\binom{N-4}{N/2-c_1'''-c_2''}\binom{N/2-c_1'''-c_2'''}{N/2-u-c_1'''}\binom{N-4-N/2+c_1'''+c_2''}{u-c_3'''}}{\binom{N}{N/2}\binom{N/2}{N/2-u}\binom{N/2}{u}}.
    \label{eq:p4}
\end{equation}

To find $p_{x_1x_2x_3x_4}'''$, each combination of $x_1,x_2,x_3,x_4$, now maps to 16 of the $2^8=256$ permutations of treatment statuses in $i,j,k,l$ for two random assignment vectors, each being equally likely. Using the same procedure as above to add these 16 probabilities in equation \eqref{eq:p4} together for each $p_{x_1x_2x_3x_4}'''$, we get
\begin{align*}
    &p_{1111}''' = p_{0000}''' = \frac{1}{\Perm{N}{4}}\left(\frac{N^4}{16} - \frac{3N^3}{4} + \frac{11N^2}{4} + N^2u - 6Nu - 3N + 2u^2 + 6u\right), \\
    &p_{1110}''' = p_{1101}''' = p_{0001}'''=p_{0010}''' = p_{0111}''' = p_{1011}''' = p_{1000}''' = p_{0100}''' =\frac{1}{\Perm{N}{4}}\left(\frac{N^4}{16} - \frac{3N^3}{8} + \frac{N^2}{2} + Nu - 2u^2\right), \\
    &p_{1001}''' = p_{0110}''' = p_{1010}''' = p_{0101}''' =\frac{1}{\Perm{N}{4}}\left(\frac{N^4}{16} - \frac{N^3}{4} + \frac{N^2}{4} - Nu + 2u^2\right), \\
    &p_{1100}''' = p_{0011}''' =  \frac{1}{\Perm{N}{4}}\left(\frac{N^4}{16} - \frac{N^3}{4} +  \frac{N^2}{4} - N^2u + 4Nu + 2u^2 - 6u\right). \\
\end{align*}
We get that equation \eqref{eq:pr_double_sum} can be written as
\begin{multline}
     E_u\left(\sum_{i=1}^{N-1}\sum_{j=i+1}^{N} \alpha_{ia}\alpha_{ib}\alpha_{ja}\alpha_{jb}\right) = \\ \frac{u^2}{\Perm{N}{2}}\sum_{i=1}^{N-1}\sum_{j=i+1}^{N}\sum_{x_1=0}^1\sum_{x_2=0}^1\sum_{x_3=0}^1\sum_{x_4=0}^1 Y_i(x_1)Y_j(x_2)Y_i(x_3)Y_j(x_4) + R_1 = \frac{u^2}{\Perm{N}{2}} S_1 + R_1.
\end{multline}
where $R_1$ is the constant and linear terms of $u$. Similarly, by symmetry, we have
\begin{multline}
    E_u\left(\sum_{i=1}^{N-2}\sum_{j=i+1}^{N-1}\sum_{k=j+1}^{N}\alpha_{ia}\alpha_{ib}\alpha_{ja}\alpha_{kb}\right) +
    E_u\left(\sum_{i=1}^{N-2}\sum_{j=i+1}^{N-1}\sum_{k=j+1}^{N}\alpha_{ia}\alpha_{ja}\alpha_{jb}\alpha_{kb}\right) + \\
    E_u\left(\sum_{i=1}^{N-2}\sum_{j=i+1}^{N-1}\sum_{k=j+1}^{N}\alpha_{ia}\alpha_{jb}\alpha_{ka}\alpha_{kb}\right) = \\ -\frac{u^2}{\Perm{N}{3}} \left(
    \sum_{i=1}^{N-2}\sum_{j=i+1}^{N-1}\sum_{k=j+1}^{N}\sum_{x_1=0}^1\sum_{x_2=0}^1\sum_{x_3=0}^1\sum_{x_4=0}^1 Y_i(x_1)Y_j(x_2)Y_i(x_3)Y_k(x_4) + \right. \\ \left.
    \sum_{i=1}^{N-2}\sum_{j=i+1}^{N-1}\sum_{k=j+1}^{N}\sum_{x_1=0}^1\sum_{x_2=0}^1\sum_{x_3=0}^1\sum_{x_4=0}^1 Y_i(x_1)Y_j(x_2)Y_j(x_3)Y_k(x_4) +  \right. \\ \left. \sum_{i=1}^{N-2}\sum_{j=i+1}^{N-1}\sum_{k=j+1}^{N}\sum_{x_1=0}^1\sum_{x_2=0}^1\sum_{x_3=0}^1\sum_{x_4=0}^1 Y_i(x_1)Y_k(x_2)Y_j(x_3)Y_k(x_4) \right) + R_2 = -\frac{u^2}{\Perm{N}{3}} S_2 + R_2.
\end{multline}

And finally, we have
\begin{multline}
     E_u\left(\sum_{i=1}^{N-3}\sum_{j=i+1}^{N-2}\sum_{k=j+1}^{N-1}\sum_{l=k+1}^N\alpha_{ia}\alpha_{ja}\alpha_{kb}\alpha_{lb}\right) = \\ 
     2\frac{u^2}{\Perm{N}{4}}\left(
     \sum_{i=1}^{N-3}\sum_{j=i+1}^{N-2}\sum_{k=j+1}^{N-1}\sum_{l=k+1}^N\sum_{x_1=0}^1\sum_{x_2=0}^1\sum_{x_3=0}^1\sum_{x_4=0}^1 Y_i(x_1)Y_j(x_2)Y_k(x_3)Y_l(x_4)+\right. \\ \left. 
     \sum_{i=1}^{N-3}\sum_{j=i+1}^{N-2}\sum_{k=j+1}^{N-1}\sum_{l=k+1}^N\sum_{x_1=0}^1\sum_{x_2=0}^1\sum_{x_3=0}^1\sum_{x_4=0}^1 Y_i(x_1)Y_k(x_2)Y_j(x_3)Y_l(x_4)+\right. \\ \left.
     \sum_{i=1}^{N-3}\sum_{j=i+1}^{N-2}\sum_{k=j+1}^{N-1}\sum_{l=k+1}^N\sum_{x_1=0}^1\sum_{x_2=0}^1\sum_{x_3=0}^1\sum_{x_4=0}^1 Y_i(x_1)Y_l(x_2)Y_j(x_3)Y_k(x_4)\right) 
     +R_3 = 2\frac{u^2}{\Perm{N}{4}}S_3 + R_3.
\end{multline}

Equation \eqref{eq:exp_iajbkalb} now becomes
\begin{equation}
    E_u\left( \sum_{i=1}^{N-1}\sum_{j=1}^{N-1}\sum_{k=i+1}^N\sum_{l=j+1}^N\alpha_{ia}\alpha_{jb}\alpha_{ka}\alpha_{lb} \right) = \frac{u^2}{\Perm{N}{2}} S_1 + R_1 -2\frac{u^2}{\Perm{N}{3}} S_2 + R_2 + 4\frac{u^2}{\Perm{N}{4}}S_3 + R_3.
\end{equation}
Inserting in equation \eqref{eq:rel_share_mse2}, defining $\psi := \frac{1}{\Perm{N}{2}}S_1 -2 \frac{1}{\Perm{N}{3}}S_2 + 4 \frac{1}{\Perm{N}{4}}S_3$ and utilizing lemma \ref{lem:bar_qu}\ref{lem:part_a}) and lemma \ref{lem:bar_qu}\ref{lem:part_b}), we get
\begin{equation}
    \var\left(\mse(\hat\tau_{\{\mathcal{W}_H\}}):\mathcal{W}_H\in \mathcal{\tilde K}_H\right)= \\
     \psi\frac{64}{N^4}\sum_{u=0}^{N/2}\left(v_u(\mathcal{\tilde{ K}}_H) - \frac{n_u}{N_A^2}\right)(u-N/4)^2.
\end{equation}
Finally, note that $v_0(\mathcal{\tilde{ K}}_H)=v_{N/2}(\mathcal{\tilde{ K}}_H)=1/H$, whereas $n_0=n_{N/2}=N_A$. Because we use $(u-N/4)^2$ instead of $u^2$, we get
\begin{equation}
    \var\left(\mse(\hat\tau_{\{\mathcal{W}_H\}}):\mathcal{W}_H\in \mathcal{\tilde K}_H\right)=\frac{64}{N^4}\psi \left( \frac{N_A-H}{HN_A}\frac{N^2}{8} + \sum_{u=1}^{N/2-1}\left(v_u(\mathcal{\tilde{ K}}_H) - \frac{n_u}{N_A^2}\right) (u - N/4)^2 \right).
\end{equation}
Using the definition of $\phi(\mathcal{\tilde{ K}}_H)$ in equation \eqref{eq:def_phi} in the paper and noting that 
\begin{equation}
    \phi(\mathcal{ K}):=\left( \frac{4}{N} \right)^2 \frac{N_A}{N_A-2} \sum_{u=1}^{N/2-1}\frac{n_u}{N_A^2}  (u - N/4)^2,
\end{equation}
we get the result in theorem \ref{thm:varvar_u}.

\clearpage

\subsection*{Simulation results for $N=100$}

\begin{figure}[h!]
     \includegraphics[width=\linewidth]{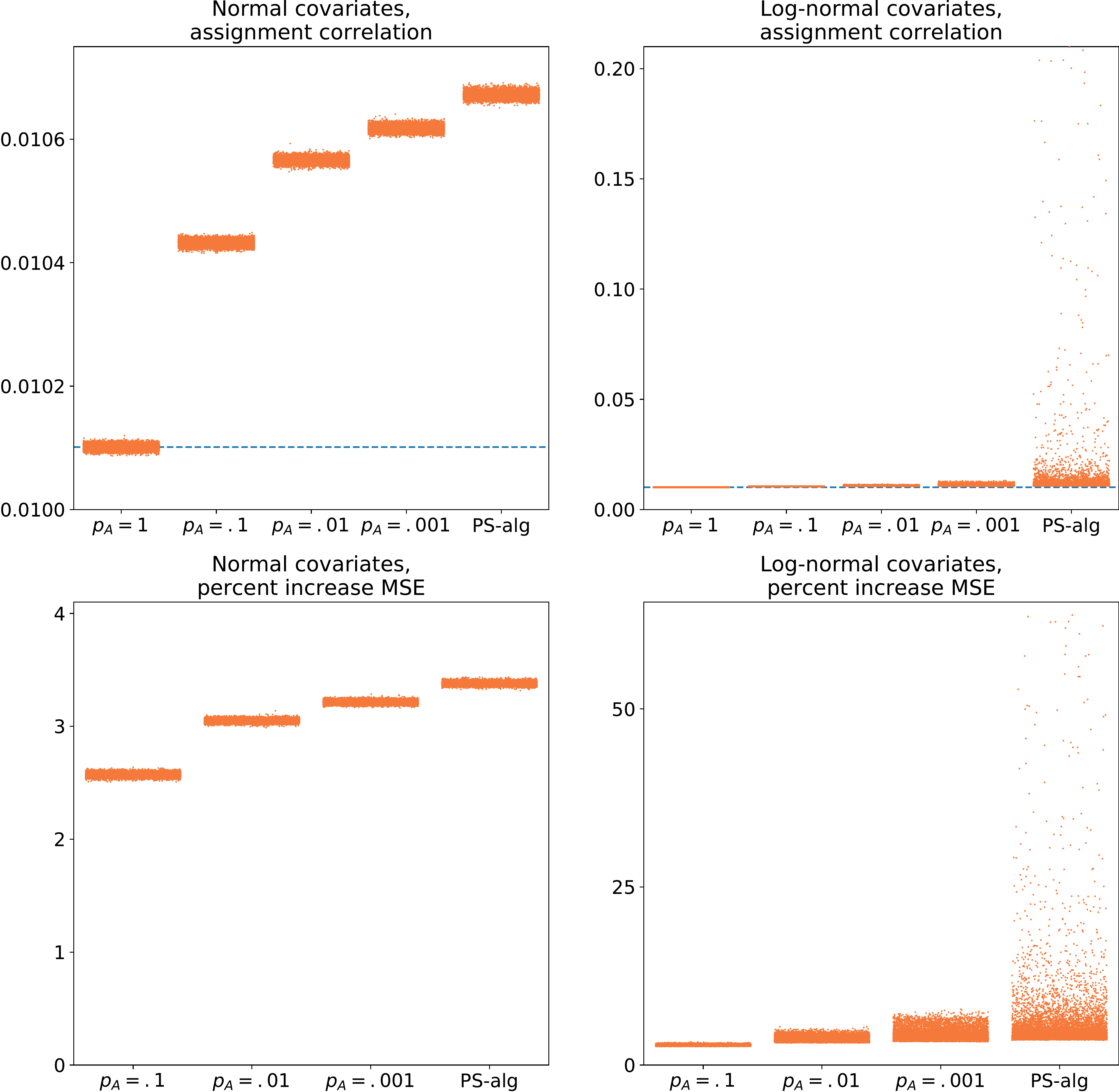}
     \caption{Distribution of assignment correlation for different designs, $N=100$ \label{fig:dist_as_cor100}}
     \floatfoot{Note: The top panel shows the distribution of the assignment correlation for different designs on normal (left graph) and lognormal (right graph) covariates with the dashed horizontal line indicating the assignment correlation for complete randomization. The bottom panel shows the corresponding MSE increase from a standard deviation increase relative to its expectation (see equation \ref{eq:rel_var_increase} in the paper). $N=100$ and for each sample and the assignment correlations are approximated with 10,000 assignment vectors (including mirrors). 10,000 random samples are drawn. Note that the scale on $y$-axis is different for each graph.}
\end{figure}

\newpage

\begin{threeparttable}[p!]
    \singlespace
    \begin{tabular*}{\textwidth}{@{\hskip\tabcolsep\extracolsep\fill}l*{1}{cccccccc}}\toprule
 & & & & \multicolumn{4}{c}{Quantiles}\\
\cmidrule(lr){5-8}& Mean & Min & Max & 0.5 & 0.75 & 0.975 & 0.999\\
\midrule\addlinespace\multicolumn{8}{l}{\textit{Five standard normal covariates}}\\
\addlinespace$p_A=1$ & 0.0101 & 0.0101 & 0.0101 & 0.0101 & 0.0101 & 0.0101 & 0.0101 \\
& & & & & & & \\ 
$p_A=0.1$ & 0.0104 & 0.0104 & 0.0104 & 0.0104 & 0.0104 & 0.0104 & 0.0104 \\
 & (2.6) & (2.5) & (2.6) & (2.6) & (2.6) & (2.6) & (2.6) \\
\addlinespace
$p_A=0.01$ & 0.0106 & 0.0105 & 0.0106 & 0.0106 & 0.0106 & 0.0106 & 0.0106 \\
 & (3.0) & (3.0) & (3.1) & (3.0) & (3.1) & (3.1) & (3.1) \\
\addlinespace
$p_A=0.001$ & 0.0106 & 0.0106 & 0.0106 & 0.0106 & 0.0106 & 0.0106 & 0.0106 \\
 & (3.2) & (3.2) & (3.3) & (3.2) & (3.2) & (3.2) & (3.3) \\
\addlinespace
PS-alg & 0.0107 & 0.0107 & 0.0107 & 0.0107 & 0.0107 & 0.0107 & 0.0107 \\
 & (3.4) & (3.3) & (3.4) & (3.4) & (3.4) & (3.4) & (3.4) \\
\addlinespace
\addlinespace\multicolumn{8}{l}{\textit{Five log-normal covariates}}\\
\addlinespace$p_A=1$ & 0.0101 & 0.0101 & 0.0101 & 0.0101 & 0.0101 & 0.0101 & 0.0101 \\
& & & & & & & \\ 
$p_A=0.1$ & 0.0105 & 0.0104 & 0.0106 & 0.0105 & 0.0105 & 0.0105 & 0.0106 \\
 & (2.8) & (2.6) & (3.2) & (2.8) & (2.9) & (3.0) & (3.1) \\
\addlinespace
$p_A=0.01$ & 0.0108 & 0.0106 & 0.0116 & 0.0107 & 0.0108 & 0.0111 & 0.0114 \\
 & (3.6) & (3.1) & (5.4) & (3.5) & (3.9) & (4.4) & (5.0) \\
\addlinespace
$p_A=0.001$ & 0.0110 & 0.0106 & 0.0132 & 0.0108 & 0.0111 & 0.0119 & 0.0128 \\
 & (4.1) & (3.3) & (7.8) & (3.8) & (4.4) & (6.0) & (7.3) \\
\addlinespace
PS-alg & 0.0123 & 0.0107 & 0.2098 & 0.0110 & 0.0113 & 0.0149 & 0.1761 \\
 & (5.2) & (3.5) & (63.2) & (4.1) & (5.0) & (9.8) & (57.6) \\
\addlinespace
\bottomrule\end{tabular*}
    \caption{Distribution of assignment correlation for different designs, $N=100$ \label{tab:dist_as_cor100}}
    \begin{tablenotes}
      \item[] \footnotesize Note: The table presents data from the same simulated distributions of the assignment correlation as is shown in Figure \ref{fig:dist_as_cor100}. The corresponding MSE increase from a standard deviation increase relative to its expectation (see equation \ref{eq:rel_var_increase} in the paper) is shown in parentheses. $N=100$ and for each sample, the assignment correlations are approximated with 10,000 assignment vectors (including mirrors). 10,000 random samples are drawn.
    \end{tablenotes}
\end{threeparttable}

\end{document}